%% file: main.tex
\documentclass[sn-mathphys-num]{sn-jnl}% Math and Physical Sciences Numbered Reference Style 
\input{macro.tex}

\begin{document}

\title[Continuous reasoning for \rev{adaptive}  container image \rev{distribution} in the cloud-edge continuum]{Continuous Reasoning for Adaptive Container Image Distribution in the Cloud-edge Continuum\footnote{This preprint has not undergone peer review or any post-submission
improvements or corrections. The Version of Record of this article is published in Springer's Cluster Computing, and
is available online at \url{https://doi.org/10.1007/s10586-025-05253-9}}}


% \author*[1,2]{\fnm{First} \sur{Author}}\email{iauthor@gmail.com}



% \affil*[1]{\orgdiv{Department}, \orgname{Organization}, \orgaddress{\street{Street}, \city{City}, \postcode{100190}, \state{State}, \country{Country}}}

% \affil[2]{\orgdiv{Department}, \orgname{Organization}, \orgaddress{\street{Street}, \city{City}, \postcode{10587}, \state{State}, \country{Country}}}

% \affil[3]{\orgdiv{Department}, \orgname{Organization}, \orgaddress{\street{Street}, \city{City}, \postcode{610101}, \state{State}, \country{Country}}}

\author[1]{Damiano Azzolini}
\equalcont{All Authors contributed equally to this study.}
\affil[1]{
    \orgdiv{Department of Environmental and Prevention Sciences},
    \orgname{University of Ferrara},
    \orgaddress{
        \city{Ferrara}, 
        \state{Italy}
    }
}

\author[2]{Stefano Forti}
\equalcont{All Authors contributed equally to this study.}
\affil[2]{
    \orgdiv{Department of Computer Science},
    \orgname{University of Pisa},
    \orgaddress{
        \city{Pisa},
        \state{Italy}
    }
}

\author[3]{Antonio Ielo}
\equalcont{All Authors contributed equally to this study. }
\affil[3]{ 
\orgdiv{Deparment of Mathematics and Computer Science}, 
\orgname{University of Calabria},%Department and Organization
\orgaddress{           
    \city{Rende},
    \state{Italy}
    }
}


\abstract{Cloud-edge computing requires applications to operate across diverse infrastructures, often triggered by cyber-physical events. Containers offer a lightweight deployment option but pulling images from central repositories can cause delays. This article presents a novel declarative approach and open-source prototype for replicating container images across the cloud-edge continuum. Considering resource availability, network QoS, and storage costs, we leverage logic programming to (\textit{i}) determine optimal initial placements via Answer Set Programming (ASP) and (\textit{ii}) adapt placements using Prolog-based continuous reasoning. We evaluate our solution through simulations, showcasing how combining ASP and Prolog continuous reasoning can balance cost optimisation and prompt decision-making in placement adaptation at increasing infrastructure sizes. }




\maketitle

\input{src/introduction}

\input{src/problemstatement}

\input{src/solutions}

\input{src/experiments}

\input{src/relatedwork}

\input{src/conclusions}

\bibliography{biblio}% common bib file

\end{document}

%% file: macro.tex
\usepackage{amssymb}
\usepackage[dvipsnames]{xcolor}
\usepackage{algorithmic}
\usepackage{graphicx}
\usepackage{subcaption}
\usepackage{xcolor}
\usepackage{tabularx}
\usepackage{verbatim}
\usepackage{xspace}
\usepackage{fancyvrb}
\usepackage{fvextra}
\usepackage{subcaption}
\usepackage{inconsolata}
\usepackage{colortbl}
\usepackage{url}
\usepackage{hyperref}
\usepackage{wrapfig}
\usepackage{placeins}
\usepackage{multirow}
\usepackage{booktabs}
\usepackage{float}
\usepackage{breqn}
\usepackage{csquotes}

\usepackage{tikz}
\usepackage{tkz-graph}
\usepackage{eurosym}
\usetikzlibrary{calc,shapes.multipart,chains,arrows}
\usetikzlibrary{arrows,automata}
\usetikzlibrary{positioning}

\usepackage[normalem]{ulem}

\usepackage{amsthm}
\newtheorem{definition}{Definition}
\newtheorem{problem}{Problem}
\newtheorem{example}{Example}
\usepackage{cleveref}

\newtheorem{theorem}{Theorem}[section]

\newtheorem{lemma}[theorem]{Lemma}

\usepackage[english]{babel}

\everymath{\displaystyle}

\crefname{section}{Section}{Section}
\Crefname{section}{Section}{Section}
\crefname{figure}{Figure}{Figure}
\Crefname{figure}{Figure}{Figure}
\crefname{table}{Table}{Table}
\Crefname{table}{Table}{Table}

\newcommand{\cd}[1]{\texttt{\small #1}\xspace}
\newcommand{\rev}[1]{\textcolor{black}{#1}\xspace}

\newcommand{\proimg}{\xspace{\sf declace}\xspace}

\newcommand{\atterm}{\textit{@}-term\xspace}

\DefineVerbatimEnvironment{code}{Verbatim}
{fontfamily=zi4, fontsize=\scriptsize, frame=single, framesep=1mm, framerule=0.1pt, rulecolor=\color{gray}}

\DefineVerbatimEnvironment{codeNum}{Verbatim}
{fontfamily=zi4, numbers=left, numbersep=5pt, numberblanklines=false, firstnumber=last, tabsize=2, fontsize=\scriptsize, frame=single, framesep=1mm, framerule=0.1pt, rulecolor=\color{gray}}

\def\BibTeX{{\rm B\kern-.05em{\sc i\kern-.025em b}\kern-.08em
    T\kern-.1667em\lower.7ex\hbox{E}\kern-.125emX}}

\usepackage{tkz-graph}
\usepackage{pgfplots}
\usepackage{subcaption} 
\usepackage{tkz-graph}

\usepackage{multirow}

\usepackage{xcolor}

\definecolor{first_color_plot}{HTML}{F30522}

\definecolor{second_color_plot}{HTML}{20D525}

\definecolor{third_color_plot}{HTML}{DED712}

\definecolor{fourth_color_plot}{HTML}{67A9CF}

\definecolor{orange_plot}{HTML}{FA5F22}
\definecolor{pink_plot}{HTML}{FFD0B4}
\definecolor{light_blue_plot}{HTML}{D1E5F0}
\definecolor{dark_blue_plot}{HTML}{2166AC}

\newcommand{\textwidthresizebox}{0.40}
\newcommand{\alignedimgsresizebox}{0.85}

\usepackage{tikz}
\usetikzlibrary{shapes.geometric, arrows}

\tikzstyle{start} = [circle,draw=black]

\tikzstyle{process} = [rectangle, 
minimum width=3cm, 
minimum height=1cm, 
text centered, 
text width=3cm, 
draw=black]

\tikzstyle{wf_node} = [rectangle, 
minimum width=1cm, 
minimum height=1cm, 
text centered, 
text width=2cm, 
draw=black,
rounded corners]

\tikzstyle{decision} = [diamond, 
minimum width=3cm, 
minimum height=1cm, 
text centered, 
text width=2cm, 
draw=black]
\tikzstyle{arrow} = [thick,->,>=stealth]

\pgfdeclarelayer{background}
\pgfdeclarelayer{foreground}
\pgfsetlayers{background,main,foreground}

%% file: src/introduction.tex
\section{Introduction}
\label{sec:intro}

Cloud-edge computing paradigms (e.g., fog~\cite{DBLP:journals/access/NahaGGJGXR18} and edge~\cite{DBLP:journals/fgcs/KhanAHYA19} computing) have been getting increasing attention both from academia and industry as a solution to deal with the considerable amounts of data coming from the Internet of Things (IoT)~\cite{DBLP:journals/fgcs/TortonesiGMRSS19}.
All such paradigms aim at placing (part of the) computation suitably closer to where IoT data are produced to aggregate and filter them along a pervasive continuum of heterogeneous infrastructure resources at the interplay between IoT devices and cloud datacentres.
Such an approach requires to deploy application components (e.g., microservices and serverless functions) in a context- and QoS-aware manner in order to improve their response times, data-locality, energy demand, security, and trust~\cite{DBLP:journals/csur/Ait-SalahtDL20}. Indeed, cloud-edge computing is particularly well-suited for latency-sensitive and bandwidth-hungry IoT applications that require prompt reactions to sensed cyber-physical events, e.g., virtual reality, remote surgery, and online gaming~\cite{DBLP:journals/logcom/FortiBB22}.

Much research (e.g., surveyed in \cite{DBLP:journals/csur/MahmudRB20}) has focussed on the decision-making process related to placing application components onto cloud-edge resources according to a large variety of constraints (e.g., hardware, IoT, and QoS) and optimisation targets (e.g., operational cost, resource usage, and battery lifetime), also in response to cyber-physical stimuli. Many proposals foresee deploying application components through containers, being a lightweight and scalable virtualisation technology that runs both on powerful cloud servers and resource-constrained edge devices~\cite{DBLP:journals/jsa/ShakaramiSGNM22,DBLP:journals/jss/SoldaniTH18}.
\rev{To run a container, one must first pull its image from the \textit{repository} stored in an \textit{image registry}, i.e., a server that acts as a distribution hub for container images, allowing users to push and pull images to and from the repository.} 

However, as highlighted in~\cite{DBLP:conf/icccn/DarrousLI19,CLOUD2023}, very few works considered the problems related to distributing container image repositories over cloud-edge \rev{registries} to further improve application performance by reducing \rev{application} start-up times. 
On one hand, downloading images from a central cloud repository is not always viable in short times, especially for devices closer to the edge of the Internet that might incur intermittent connectivity or scarce bandwidth availability. On the other hand, \rev{some} edge application components require fast start-ups, being deployed to the task in response to cyber-physical events that trigger them, or to specific end-users mobility patterns. This is especially true in \rev{Function-as-a-Service} edge settings, where short-running containerised stateless functions run only when needed~\cite{DBLP:conf/acsw/AslanpourTCJSTA21}.

Hereinafter, we precisely consider the possibility of distributing container images across cloud-edge \rev{registry} nodes to speed up on-demand image downloads to any other node and, therefore, containers' start-up times.  In the above settings, deciding onto which infrastructure nodes to place different container images is an interesting and challenging problem~\cite{DBLP:conf/icccn/DarrousLI19,CLOUD2023}. Image placement should guarantee that images can be downloaded from any infrastructure node within a specific time frame. 
In doing so, it should replicate images onto a subset of infrastructure nodes in a QoS-aware and context-aware manner, without exceeding the capacity of each node and containing operational costs related to leasing the needed storage resources. Indeed, if relying on a central cloud repository might suffer from too high transfer times, replicating all images on all available nodes might be economically unsustainable or unfeasible due to edge resource constraints.

The above scenario is characterised by large-scale deployments, high network churn, and frequent image repository updates. Decision-making should scale to quickly distribute tens of images over hundreds of cloud-edge \rev{registry} nodes, encompassing their high heterogeneity. Edge nodes can temporarily join and leave the network either due to failures (e.g., power disruption, node crash, reboot) or to network QoS degradation (e.g., network congestion, increased jitter due to node mobility)~\cite{DBLP:journals/fgcs/FortiGB21}.

Image repositories are continuously updated by developers as soon as new features are released or newly discovered bugs are fixed. As for 2018, operating system (OS) images such as Ubuntu and Alpine received monthly updates, while non-OS images such as Nginx and MongoDB received up to ten updates per month~\cite{dockerupdates}.
Hence, decision-making should run continuously to adapt image placements in response to changes in the infrastructure (e.g., node/link crashes, link QoS degradation) or in the image repository to be distributed (e.g., image addition/removal, increase of storage requirements), that may \rev{negatively affect image transfer times.} %compromise their effectiveness in reducing image transfer times.  
\rev{Hence,} adaptation should also migrate as few images as possible to avoid unnecessary data transfers. To the best of our knowledge,
\textit{no prior work exists tackling these adaptation aspects}.

In this article, we propose a novel declarative methodology -- and its open-source prototype\footnote{Freely available at \url{https://github.com/teto1992/declace}.} \proimg{} -- to solve the problem of adaptively distributing base container images along a cloud-edge infrastructure in a QoS-, context- and cost-aware manner. 
Our prototype extensively relies on \textit{logic programming}. It combines a cost-optimal image placement strategy, written in Answer Set Programming (ASP), to determine initial image placements and a heuristic iterative deepening placement strategy, written in Prolog, to adapt such initial placements via continuous reasoning\footnote{Originally proposed at Meta, continuous reasoning speeds up static analyses over Facebook's shared codebase by only focussing on incrementally analysing the latest changes~\cite{DBLP:journals/cacm/DistefanoFLO19}. More recently, Forti et al.~\cite{DBLP:journals/logcom/FortiBB22} and Herrera et al.~\cite{DBLP:journals/computing/HerreraBFBM23} extended the benefits of continuous reasoning to cloud-edge application management, adapting application deployments by possibly migrating only services in need for attention due to infrastructure or application changes.} when needed. Thanks to the declarative nature of logic programming, \proimg{} is more \textit{concise} ( $\simeq100$ lines of code), \textit{readable}, and \textit{extensible} compared to imperative solutions.  %Additionally, it is extensible to accommodate new requirements in ever-changing cloud-edge settings and intrinsically \textit{explainable} as obtained image placements are determined as proofs over a formal logic theory.

While initial placements tolerate the longer decision-making times of our ASP approach in favour of minimising operational placement costs (i.e., it solves the problem by finding the \textit{global} minimum), runtime adaptation calls for prompt decision-making to minimise disruption of the image distribution service, even at the price of sub-optimal placements.
By employing heuristic \textit{continuous reasoning} to solve the container image placement problem at runtime, this work enables the 
%continuous QoS-, context- and cost-aware 
adaptation of distributed image repositories to changing cloud-edge environments and image requirements. 
As we will show through our experiments, continuous reasoning reduces the size of the problem instances to be solved at runtime and contains the number of images to be migrated when adapting to changing conditions. This permits handling various infrastructure sizes and speeds up decision-making.

The rest of this article is organised as follows. After formally defining the decision and optimisation problems of interest for this work (Section~\ref{sec:problemstatement}), we detail the Prolog and ASP solutions to such problems and how they are combined into the \proimg{} prototype tool (Section~\ref{sec:methodology}). We then describe the results of the experimental assessment of \proimg{} through simulation over lifelike data (Section~\ref{sec:experiments}). Finally, after discussing some \rev{closely} related work (Section~\ref{sec:related}), we conclude by pointing to directions for future work (Section~\ref{sec:conclusions}).

%% file: src/problemstatement.tex
\section{Considered Problem}
\label{sec:problemstatement}

\rev{In this section, we illustrate a formal model of the considered cloud-edge settings and formally define the problems we will tackle, while proving their NP-hard nature. }
%
%We also introduce a running example that will be used throughout the article. %to illustrate our proposal.

\subsection{Formal Model}
\label{sec:formalmodel}

\noindent
First, we model \rev{base} container images\footnote{\rev{We here focus on \textit{base images} as they are needed to define and build any other container image. %Modelling of layered images is in the scope for future work. 
}} as follows.

\smallskip\begin{definition}\label{def:image}
A {\em container image} is a triple $\langle i, s, m \rangle$ where $i$ denotes the unique image identifier, $s$ the size of the image in Megabytes (MB), and $m$ the maximum amount of time in seconds tolerated to download such image onto a target node. 
\end{definition}

 \smallskip\begin{example}\label{ex:images} The latest version\footnote{Image sizes are taken from Docker by pulling the images with \cd{latest} tag as of June 2024.} of the \cd{alpine}, \cd{ubuntu} and \cd{nginx} images are modeled by the following triples: %can be modelled as follows
%\damiano{TODO: aggiungerei un timestamp visto che usiamo latest version, DONE}} 
%
$$\langle \cd{alpine}, 8\ \textrm{\em MB}, 30\ \textrm{\em s} \rangle, \langle \cd{ubuntu}, 69\ \textrm{\em MB}, 60\ \textrm{\em s} \rangle, \langle \cd{nginx}, 192\ \textrm{\em MB}, 120\ \textrm{\em s} \rangle$$

\noindent enforcing that \cd{alpine} can be downloaded from any network node within 30 seconds, \cd{ubuntu} within 60 seconds, and \cd{nginx} within 120 seconds.\hfill $\qed$
\end{example}
\smallskip

We model registry nodes, their available storage and storage cost, and the Quality of Service of end-to-end links between them in terms of latency and bandwidth. \rev{We assume that the considered nodes host a container image registry, i.e., they can be used to distributed container image repositories across the cloud-edge network.}

\smallskip
\smallskip\begin{definition}\label{def:node}
A {\em registry node} is a triple $\langle n, h, c \rangle$ where $n$ denotes the unique node identifier (e.g., its IP address), $h$ the available amount of storage it features expressed in MB, and $c$ the monthly cost per used unit of storage (MB) at such node.
\end{definition}

\smallskip
\smallskip\begin{definition}\label{def:link}
A {\em link} is a 4-tuple $\langle n, n', \ell, b \rangle$ where $\ell$ and $b$ denote the end-to-end latency in milliseconds and bandwidth in Mbps, respectively, between the nodes identified by $n$ and $n'$.
\end{definition}

 \smallskip\begin{example}
\label{ex:running_example}

Consider the cloud-edge infrastructure of Figure~\ref{fig:exampleInfrastructure}, composed of one cloud and five edge nodes. As an example, the \cd{cloud} node can be modelled via the tuple
$$\langle \cd{cloud}, 2048\ \textrm{\em MB},  0.9\ \text{\euro/MB} \rangle$$

For the sake of readability, Figure~\ref{fig:exampleInfrastructure} only shows direct (wired or wireless) links between nodes. We assume that all other end-to-end links are computed by selecting shortest paths based on latency, and that the end-to-end bandwidth corresponds to the minimum bandwidth\footnote{It is worth noting that, even though we assume symmetric links in this example, our solution also works with asymmetric bandwidth and latency combinations between pairs of nodes.} along the path. Then, the links between \cd{cloud} and \cd{edge1} and between \cd{cloud} and \cd{edge5} can be modelled as follows
$$\langle \cd{cloud}, \cd{edge1}, 20\ \textrm{\em ms},  50\ \textrm{\em Mbps} \rangle, \langle \cd{cloud}, \cd{edge5}, 55\ \textrm{\em ms},  5\ \textrm{\em Mbps} \rangle$$

\noindent with traffic from \cd{cloud} to \cd{edge5} routed through \cd{edge2}. \hfill$\qed$

\end{example}

\input{img/infrastructure}

% \begin{figure}[!h]
%     \centering
%     \includegraphics[width=0.49\textwidth]{img/simpleInfra.pdf}
%     \caption{Example infrastructure.}
%     \label{fig:exampleInfrastructure}
% \end{figure}

It is \rev{now} possible to formally define image placements. Note that one image can be stored onto multiple nodes, i.e., we foresee the possibility to replicate an image onto multiple infrastructure nodes.

\smallskip\begin{definition}
Let $I$ be a set of container images and $N$ be a set of available infrastructure nodes. An {\em image placement} is a mapping $p: I' \rightarrow \mathcal{P}(N)$ that maps each image of $I' \subseteq I$ to a subset\footnote{$\mathcal{P}(N)$ denotes the power set of $N$, i.e., the set of all possible subsets of $N$.} of infrastructure nodes $\overline{N}\subseteq N$. If $I'=I$, then the placement is {\em total}, it is {\em partial} otherwise.
\end{definition}

 \smallskip\begin{example}\label{ex:placement}
Considering the images of Example~\ref{ex:images} as $I$ and the infrastructure of Example~\ref{ex:running_example} as $N$ and $L$, a possible \emph{partial} placement $p'$ is defined by:

$$p'(\cd{alpine}) = \{\cd{cloud, edge2}\} \ \ \ \ p'(\cd{ubuntu}) = \{\cd{edge5}\} $$

\smallskip\noindent as a placement is not defined for image \cd{nginx}. %\damiano{TODO: $<-$ Non ho capito questa frase}. 
Thus, a \emph{total} placement $p$ is defined by:
% P = [at(alpine, n2), at(ubuntu, n5), at(ubuntu, n2), at(nginx, n3), at(nginx, n5), at(nginx, n2)],
$$p(\cd{alpine}) = \{\cd{edge2}\} \ \ p(\cd{ubuntu}) = \{\cd{edge2, edge5}\} \ \ p(\cd{nginx}) = \{\cd{edge2, edge3, edge5}\} $$

\noindent since it maps all images onto a target node.\hfill$\qed$
\end{example}\smallskip

Assuming a cap on the number of replicas for all images, we \rev{now} define the eligibility criteria for image placement. An eligible image placement must (1)~not exceed the cap on the number of replicas for any placed image, (2)~enable all infrastructure nodes to download any image within its associated maximum time frame (including transmission and propagation times) or to have it already stored, and (3)~not exceed storage capacity at any infrastructure node. These conditions correspond to the three conditions in the following definition of eligible placement.

\smallskip\begin{definition}\label{def:eligiblePlacement}
Let $I$ be a set of container images, $N$ be a set of available infrastructure nodes, $p: I' \rightarrow \mathcal{P}(N)$ be an image placement over $I'\subseteq I$, $R \in \mathbb{N}$ be a maximum image replica factor, and $L$ be the set of links between nodes in $N$. The image placement $p$ is an {\em eligible placement} over $I'$ if and only if:
\begin{equation}\label{eq:maxreplicas}
    \forall i \in I'. \quad |p(i)| \leq R
\end{equation}
\begin{align}\label{eq:transfertimes}
\forall \langle n, &\_, \_ \rangle  \in  N.\  \forall  \langle  i, s, m \rangle  \in  I'.\ \exists \ \langle n',\_,\_\rangle \in p(i)\ : \nonumber \\ 
  & n=n'\ \lor\ ( \langle n', n, \ell, b \rangle \in L\ \wedge\ 8\times s/b + \ell/1000 \leqslant m )
  \end{align}
\begin{align}\label{eq:hardwarelimit}
\forall \langle n, h, \_ \rangle \in N.\ I'_n = \{\  \langle i, s, \_ \rangle \in I'\ |\ p(i) &= n\ \} \Rightarrow \sum_{\langle \_, s,\_\rangle \in I'_n}s \leqslant h
  \end{align}
\end{definition}

Note that Definition~\ref{def:eligiblePlacement} is incremental and enables reasoning on partial eligible image placements while extending them to total eligible placements. Such an aspect is crucial for continuous reasoning and will be key to the Prolog code presented next. %in Section~\ref{sec:continuousreasoning}.
The $8$ constant factor in Equation~\ref{eq:transfertimes} is needed to convert $b$'s Mbps to MB/ms and correctly sum them with the milliseconds of $\ell$, converted to seconds by dividing by 1000.

Eligible placements can then be associated with their monthly cost as follows.%according to the following definition. 

\smallskip\begin{definition}
Let $I$ be a set of container images, $N$ be a set of available infrastructure nodes, and $p: I \rightarrow \mathcal{P}(N)$ be an eligible image placement. The {\em cost} of $p$ is defined as:

$$cost(p) = \sum_{\langle \_, s, \_ \rangle \in I}\ \sum_{\langle \_, \_, c\rangle \in p(i)} s \times c$$
\end{definition}

Lastly, we provide a notion of optimality of an eligible placement based on its associated cost.

\smallskip\begin{definition}\label{def:costOptimal}
An eligible image placement $p: I \rightarrow \mathcal{P}(N)$ is {\em cost-optimal} if no other eligible image placement $p': I \rightarrow \mathcal{P}(N)$ exists such that $p \neq p' \ \land\ cost(p') < cost(p)$.
\end{definition}

 \smallskip\begin{example}\label{ex:eligibleplacement}
Retaking Example~\ref{ex:placement} and setting $R=3$, the first partial placement $p'$ is not eligible as, for instance, it is not possible to transfer image \cd{ubuntu} from node \cd{edge5} to node \cd{cloud} within 60 seconds (as it requires over 110 seconds). Readers can instead verify that the total placement $p$ of Example~\ref{ex:placement} is eligible as it meets all conditions {\em(1)}, {\em(2)}, and {\em(3)} of Definition~\ref{def:eligiblePlacement}. Besides, as we will programmatically check in Section~\ref{sec:methodology}, placement $p$, with a cost of \euro 308, is also cost-optimal in the described scenario. \hfill$\qed$
\end{example}

\subsection{Problem Statements and Complexity}
\label{sec:consideredproblems}
In this section we formally define the problems we address in this paper.
The first (decision) problem we consider in this article is precisely how to determine eligible image placements in Cloud-Edge landscapes. We define it as follows.

\smallskip
\begin{problem}[Image Placement, IPP]
\label{problem:image_placement}
An instance of the image placement problem is a 4-tuple $\langle N, L, I, R \rangle$ where 
$N$ is a set of available infrastructure nodes, 
$L$ is the set of end-to-end links between such nodes, 
$I$ is a set of container images, and 
$R$ is a maximum replica factor. Determine (if it exists) a total eligible placement $p: I \rightarrow \mathcal{P}(N)$
 as per {\em Definition~\ref{def:eligiblePlacement}}.
\end{problem}

\smallskip
%%%%%%%% SNIP SNIP SNIP
We prove that Problem~\ref{problem:image_placement} is NP-hard by sketching a polynomial time reduction of the NP-hard~\cite{DBLP:books/fm/GareyJ79} problem of bin-packing (BPP) to IPP. Bin packing can be stated as follows. 
Let $O$ be a finite set of items each of size $s(o) \in \mathbb{N}$. Determine if there exists a partition of items in $O$ over $K\in \mathbb{N}$ bins without exceeding their capacity $B\in \mathbb{N}$. 

\smallskip
\begin{lemma}
\label{th:IPP_complexity}
IPP $\in$ NP-hard. 
\end{lemma}

\begin{proof}
Given an instance $\langle O, K, B \rangle$ of BPP it is possible to reduce it to an instance $\langle N, L, I, R \rangle$ of IPP as follows:
\begin{enumerate}
    \item populate $N$ with $K$ nodes of the form $\langle k, B, 0 \rangle$ with $k = 1,2,\dots,K$, size $B$, and null cost,
    \item populate $L$ with a complete set of end-to-end links between pairs of nodes $i,j \in N$ of the form $\langle i, j, 0, \infty \rangle$, so that nodes reach each other with null latency and infinite bandwidth, 
    \item populate $I$ with a container image $\langle o, s(o), \infty \rangle \in I$ for each object $o \in O$, where the size of the object becomes the size of the image and the maximum tolerated transfer time is infinite, and
    \item set the maximum replica factor $R$ to $1$, i.e. placing each image to a single node only.
\end{enumerate}
\noindent
Determining a solution to the above instance of IPP would partition images in $I$ (i.e., item in $O$) over $K$ nodes (i.e., bins) without exceeding their capacity $B$. Thus, it would solve the original instance of BPP.
The above steps reduce an instance of BPP to an instance of IPP, incurring a polynomial time complexity of $O(K^2)$, due to step 2, which requires creating a link for each pair of nodes in $N$.
%Disposing of a polynomial time oracle for IPP would, in turn, allow solving BPP in polynomial time thanks to the above reduction. 
This proves that IPP is at least as difficult as BPP, which is known to be NP-hard. Hence, IPP is also NP-hard. 
\end{proof}

Such a result means that no poly-time algorithm is currently known to solve IPP (unless $P=NP$ is proven). However, checking whether a placement $p$ is eligible, according to Definition~\ref{def:eligiblePlacement}, can be performed naively in $O(|I|N^{2})$. This consideration along with the previous lemma proves that IPP $\in$ NP-complete.
The second problem we will be studying is to continuously adapt a previously eligible total placement whenever changes in the underlying network topology or image requirements occur that affect its eligibility. 

\smallskip
\begin{problem}[Continuous Image Placement, CIPP]
\label{problem:continuous_image_placement}
Let $\langle N,L,I,R\rangle$ and $\langle N',L',I',R' \rangle$ be two image placement problem instances and $p: I \rightarrow \mathcal{P}(N)$ be a total eligible placement for the first problem instance. Let $I_{OK}$ be a maximal proper subset of $I'$ such that $p$ is eligible over $I_{OK}$, and $I_{KO} = I' \setminus I_{OK}$. Determine an eligible total placement $p': I' \rightarrow \mathcal{P}(N')$ (if any) such that $p'(i) = p(i), \forall i \in I_{OK}$. 
\end{problem}\smallskip

Said otherwise, given infrastructure or image requirements changes that make a previously total eligible placement $p$ only partially eligible, we aim at extending $p$ over $I_{OK}$ into a total eligible placement $p'$ over $I$, by (possibly) only modifying the placement of images in $I_{KO}$.

Last, following Definition~\ref{def:costOptimal}, we can easily define the optimisation problem corresponding to Problem~\ref{problem:image_placement}.

\smallskip
\begin{problem}[Optimal Image Placement, OIPP]
\label{problem:optimal_image_placement}
Let $\langle N,L,I,R\rangle$ be a placement problem instance, determine whether it exists an eligible cost-optimal image placement $p: I \rightarrow \mathcal{P}(N)$
 as per {\em Definition~\ref{def:eligiblePlacement}} and {\em Definition~\ref{def:costOptimal}}.
\end{problem}\smallskip

\noindent As solving CIPP and OIPP implies solving IPP, Problems \ref{problem:continuous_image_placement}  and \ref{problem:optimal_image_placement} are also NP-hard.

%% file: img/infrastructure.tex
% \begin{adjustbox}{width=.5\textwidth,center}
\begin{figure}[thb]
\centering
\begin{tikzpicture}[
% ->,
>=stealth',
shorten >=1pt,
% auto,
node distance=3cm,
semithick,
every text node part/.style={align=center},
scale=0.66
% every edge/.style = {draw, semithick},
% every edge quotes/.style={fill=white,font=\footnotesize}
]

  \node[state] (cloud) at (0,9) {{\footnotesize \textbf{cloud}} \\ {\footnotesize 1 TB} \\ {\footnotesize \euro 0.7}};
  \node[state] (edge1) at (-3, 5.5)  {\footnotesize  \textbf{edge1} \\ \footnotesize  64 GB \\ \footnotesize  \euro 0.2};
  \node[state] (edge2) at (3, 5.5) {\footnotesize  \textbf{edge2} \\ \footnotesize  32 GB \\ \footnotesize \euro 0.7};
  \node[state] (edge5) at (3, 2) {\footnotesize  \textbf{edge5} \\ \footnotesize  2 GB \\ \footnotesize \euro 0.9};
  \node[state] (edge3) at (-5,2) {\footnotesize  \textbf{edge3} \\ \footnotesize  8 GB \\ \footnotesize \euro 0.4};
  \node[state] (edge4) at (-1,2) {\footnotesize  \textbf{edge4} \\ \footnotesize  4 GB \\ \footnotesize \euro 0.6};

  \draw (cloud) edge node[above,sloped] {\footnotesize 50 Mbps\\\footnotesize 20 ms} (edge1)
        (cloud) edge node[below,sloped] {\footnotesize 70 Mbps\\ \footnotesize 25 ms} (edge2)
        (edge1) edge node[below,sloped] {\footnotesize 100 Mbps\\\footnotesize 5 ms}(edge2)
        (edge1) edge node[above,sloped] {\footnotesize 10 Mbps\\ \footnotesize 10 ms}(edge3)
        (edge1) edge node[below,sloped] {\footnotesize 30 Mbps\\ \footnotesize 10 ms} (edge4)
        (edge2) edge node[right] {\footnotesize 5 Mbps\\ \footnotesize 30 ms} (edge5);
\end{tikzpicture}
\medskip
\caption{Example infrastructure.}
\label{fig:exampleInfrastructure}
\end{figure}
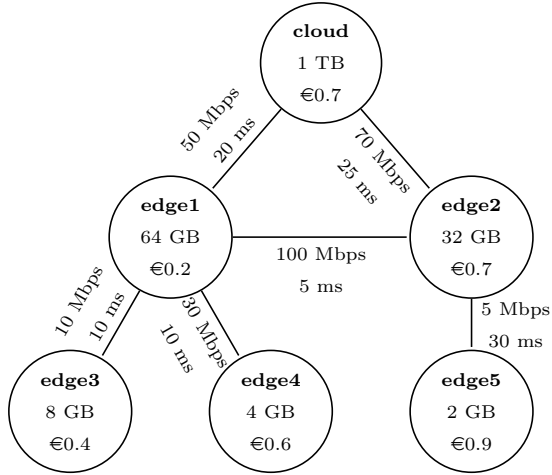
% \end{adjustbox}

%% file: src/solutions.tex
\section{Methodology}
\label{sec:methodology}

\rev{In this section we propose declarative solutions to all the problems introduced in Section~\ref{sec:kr}. In particular, we will detail a knowledge representation for our problems, two Prolog iterative deepening solutions to solve IPP and CIPP, respectively, and an ASP encoding to efficiently solve OIPP. Summing up, we show how to combine Prolog continuous reasoning with the ASP encoding to implement adaptive QoS-, context- and cost-aware container image distribution in cloud-edge landscapes.} We assume readers to be familiar with basics of logic programming. For a complete treatment of the Prolog and ASP, we refer them to~\cite{sterling1994art} and~\cite{DBLP:books/sp/Lifschitz19}.

\subsection{Knowledge Representation}
\label{sec:kr}

We encode a problem instance as a set of facts over the predicates \cd{node/3}, \cd{image/3}, and \cd{link/4}, which can be mapped one to one to Definitions~\ref{def:image}, \ref{def:node}, and \ref{def:link} of Section~\ref{sec:problemstatement}, respectively. 
In particular, container images are denoted by facts like \cd{image(ImgId, Size, Max)},
%
%\medskip\begin{code}
%image(ImgId, Size, Max).
%\end{code}
%
\noindent where \cd{ImgId} is a unique image identifier, \cd{Size} is the image size expressed in MB, and \cd{Max} is the maximum tolerated download latency for the image at hand expressed in ms.
Similarly, infrastructure nodes are declared as facts of the form \cd{node(NodeId, Storage, Cost)}
%
%\medskip\begin{code}
%node(NodeId, Storage, Cost).
%\end{code}
%
\noindent where \cd{NodeId} is a unique node identifier, \cd{Storage} the available node storage capacity expressed in MB, and \cd{Cost} the unit storage cost per MB.
Links between nodes are declared as facts of the form \cd{link(NodeId1, NodeId2, Latency, Bandwidth)}
%
%\medskip\begin{code}
%link(NodeId1, NodeId2, Latency, Bandwidth).
%\end{code}
%
\noindent where \cd{NodeId1} and \cd{NodeId2} are the considered link endpoints, and \cd{Latency} and \cd{Bandwidth} the average end-to-end latency in ms and bandwidth in Mbps featured by such a link.
Finally, the maximum image replica factor is denoted by \cd{maxReplicas(R)}
%
%\medskip\begin{code}
%maxReplicas(R).
%\end{code}
%
where \cd{R} is an integer.

\medskip\begin{example}\label{ex:encoding}
A portion of the problem instance described in Examples~\ref{ex:images}--~\ref{ex:eligibleplacement} and sketched in Figure~\ref{fig:exampleInfrastructure} is encoded as follows:

\medskip\begin{code}
image(alpine,8,30).          image(ubuntu,69,60).       image(nginx,192,120).

node(cloud,1024000,0.7).     node(edge1,64000,0.7).     node(edge2,32000,0.4).
node(edge3,8000,0.5).        node(edge4,4000,0.7).      node(edge5,2000,0.4).

link(edge0,edge1,20,50).     link(edge1,edge0,20,50).    link(edge0,edge2,25,70).    
link(edge2,edge0,25,70).     link(edge1,edge2,5,100).    link(edge2,edge1,5,100).    
link(edge1,edge3,15,10).     link(edge3,edge1,15,10).    link(edge1,edge4,10,30).
link(edge4,edge1,10,30).     link(edge2,edge5,30,5).     link(edge5,edge2,30,5).

maxReplicas(3).
\end{code}

For the sake of readability, we only reported the direct links sketched in Figure~\ref{fig:exampleInfrastructure}. 
%\damiano{All other \cd{link/4} facts, computed as described in Example~\ref{ex:running_example}, are listed in the online repository. -- toglierei} \hfill$\qed$
%\damiano{qui mettere nel repo un file running example con questo -- promemoria}

\end{example}

\subsection{Solving IPP with Prolog}
\label{sec:iterativedeepeningprolog}

Let us now discuss the Prolog encoding to solve the IPP (Problem~\ref{problem:image_placement}) task.

\smallskip\noindent\textbf{Placement Strategy} Given a (non-empty) list \cd{[I|Is]} of images to be placed, a list of candidate placement \cd{Nodes}, a maximum replica factor \cd{R}, and a partial eligible \cd{PPlacement}, predicate \cd{imagePlacement/5} (lines 1--4) determines a total eligible \cd{Placement} by relying on subsequent evaluations of predicate \cd{replicaPlacement/5} (line 2). While recurring on the image list, \cd{replicaPlacement/5} extends the previous partial placement \cd{PPlacement} into a new eligible partial placement \cd{TmpPPlacement} which also includes image \cd{I} at hand (line 3). Recursion ends when the list of images to be placed has been fully consumed, viz., it is empty (line 4).

\medskip\begin{codeNum}[firstnumber=1]
imagePlacement([I|Is], Nodes, PPlacement, Placement, R) :-
    replicaPlacement(I, Nodes, PPlacement, TmpPPlacement, R), 
    imagePlacement(Is, Nodes, TmpPPlacement, Placement, R).
imagePlacement([], _, Placement, Placement, _).
\end{codeNum}

In turn, given an image \cd{I}, list of \cd{Nodes}, a maximum replica factor \cd{R} and a partial eligible placement \cd{PPlacement}, predicate \cd{replicaPlacement/5} (lines 5--13) extends \cd{PPlacement} into a new partial eligible \cd{NewPPlacement} by adding at most \cd{R} replicas of image \cd{I}. Until the maximum replica factor has not been exceeded (line~6, viz., Equation~\ref{eq:maxreplicas} of Definition~\ref{eq:maxreplicas}) and the transfer time condition for image \cd{I} in \cd{PPlacement} does not hold (line~7, viz., Equation~\ref{eq:transfertimes} of Definition~\ref{def:eligiblePlacement}), \cd{replicaPlacement/5} selects a candidate deployment node \cd{N} (line 8) onto which \cd{I} has not been placed yet (line 9). It then checks whether node \cd{N} can support all images that it hosts as per \cd{PPlacement} and the new replica of image \cd{I} (line~10, viz., Equation~\ref{eq:hardwarelimit} of Definition~\ref{def:eligiblePlacement}). If this holds, the association \cd{at(I,N)} is added to \cd{PPlacement} in the recursive call (line 11). The recursion ends when the transfer times condition holds for image \cd{I} by returning the newly constructed eligible partial \cd{Placement} (line~13). 

\medskip\begin{codeNum}[firstnumber=5, breaklines]
replicaPlacement(I, Nodes, PPlacement, NewPPlacement, R) :-
    R > 0, NewR is R - 1,
    \+ transferTimesOk(I, Nodes, PPlacement), 
    image(I, Size, _), member(N, Nodes),
    \+ member(at(I,N), PPlacement),
    storageOk(PPlacement, N, Size),
    replicaPlacement(I, Nodes, [at(I, N)|PPlacement], NewPPlacement, NewR).
replicaPlacement(I, Nodes, Placement, Placement, _) :- 
    transferTimesOk(I, Nodes, Placement).
\end{codeNum}

Predicates \cd{transferTimesOk/3} and \cd{storageOk/3} are in charge of checking conditions imposed by Equation~\ref{eq:transfertimes} and Equation~\ref{eq:hardwarelimit} of Definition~\ref{def:eligiblePlacement}, respectively.

Given an image \cd{I}, a list of nodes \cd{[N|Ns]}, \cd{transferTimesOk/3} (lines 14--20) recursively checks that for each node in the list there exists a node \cd{M} (possibly \cd{M}=\cd{N}) from which it is possible to download \cd{I} within a transfer time of at most \cd{MaxT} (lines 17--18), as specified in the corresponding\cd{image/3} fact (line 16). Predicate \cd{transferTime/4} (line 19--24) simply computes the transfer time \cd{T} of an \cd{Image} from a source node \cd{Src} to a destination node \cd{Dst} following Equation~\ref{eq:transfertimes} of Definition~\ref{def:eligiblePlacement}. We use here the \textit{cut} operator \cd{!} (line 16). Such an operator prevents unwanted backtracking, thus avoiding determining unneeded solutions. Indeed, as per Equation~\ref{eq:transfertimes} of Definition~\ref{ex:eligibleplacement}, the availability of one image source is sufficient for each destination node.

\medskip\begin{codeNum}[firstnumber=14, breaklines, commandchars=\\\{\}]
transferTimesOk(I, [N|Ns], P) :-
    dif(P,[]), member(at(I,M),P), image(I,_,MaxT), 
    transferTime(I,M,N,T), T < MaxT, !, \textcolor{olive}{% one source is enough}
    transferTimesOk(I, Ns, P).
transferTimesOk(_, [], _).

transferTime(Image, Src, Dst, T) :-
    dif(Src, Dst), node(Src, _, _), node(Dst, _, _),
    image(Image, Size, _),
    link(Src, Dst, Latency, Bandwidth),
    T is Size * 8 / Bandwidth + Latency / 1000.
transferTime(_, N, N, 0).
\end{codeNum}

Given a partial eligible \cd{Placement}, a node \cd{N} and the \cd{Size} of a new image to be placed on \cd{N}, predicate \cd{storageOk/3} checks whether node \cd{N} can also host the new image without exceeding its capacity. It first retrieves the node \cd{Storage} capacity (line~26) and sums up resources allocated at \cd{N} by the partial placement through predicate \cd{usedHW/3} (line~27, 29--30). Finally, it checks Equation~\ref{eq:hardwarelimit} of Definition~\ref{def:eligiblePlacement} for the partial \cd{Placement} extended with the new image replica (line 28).

\medskip\begin{codeNum}[firstnumber=25, breaklines]
storageOk(Placement, N, Size) :- 
    node(N, Storage, _),
    usedHw(Placement, N, UsedHw),
    Storage - UsedHw >= Size.

usedHw(P, N, TotUsed) :- 
    findall(S, (member(at(I,N), P), image(I,S,_)), Used), sum_list(Used, TotUsed).
\end{codeNum}

The above 30 lines of code, mimicking the incremental nature of Definition~\ref{def:eligiblePlacement}, already solve the \textit{image placement problem} (viz., Problem~\ref{problem:image_placement}) stated in Section~\ref{sec:problemstatement} by determining a total eligible placement of a set of images onto a target infrastructure.
%\damiano{il codice è rimasto uguale?}

\medskip\begin{example}\label{ex:imagePlacement}
Querying the \cd{imagePlacement/5} predicate over the input of Example~\ref{ex:encoding} returns as a first result the eligible image placement \cd{P}
\medskip\begin{code}[breaklines]
?- imagePlacement([alpine,ubuntu,nginx], [cloud,edge1,edge2,edge3,edge4,edge5], [], P, 3).
P = [at(nginx, edge5), at(nginx, edge3), at(nginx, cloud), at(ubuntu, edge5), at(ubuntu, edge1), at(ubuntu, cloud), at(alpine, cloud)].
\end{code}
\noindent which, going back to the model of Section~\ref{sec:formalmodel}, corresponds to the complete eligible placement $p$ defined as follows
$$p(\cd{alpine:latest}) = \{\cd{cloud}\} $$
$$p(\cd{ubuntu:latest}) = \{\cd{cloud, edge1, edge5}\}\ \ \ p(\cd{nginx:latest}) = \{\cd{cloud, edge3, edge5}\} $$

\hfill $\qed$

\end{example}

\noindent\textbf{Iterative Deepening} To limit the amount of memory used by Prolog's backtracking and make our program more efficient, \cd{imagePlacement/5} is wrapped in the \cd{id/6} predicate (lines 31--36) that performs iterative deepening~\cite{AIMA}. At each step, if the replica factor \cd{R} has not reached the maximum replica factor \cd{RMax} (line 32), \cd{id/6} exploits \cd{imagePlacement/5} to determine a total eligible \cd{Placement} with at most \cd{R} replicas of each image (line 33). If such a placement is not found, the value of \cd{R} is incremented (line 35) and \cd{id/6} is called recursively (line 36). 

\medskip\begin{codeNum}[firstnumber=31, breaklines, commandchars=\\\{\} ]  
id(Images, Nodes, PP, Placement, R, RMax) :-
    R =< RMax, 
    imagePlacement(Images, Nodes, PP, Placement, R), !.
id(Images, Nodes, PP, Placement, RMax) :-
    R =< RMax, NewR is R+1,
    id(Images, Nodes, PP, Placement, NewR, RMax).
\end{codeNum}

The \cd{placement/6} (lines 37--41) predicate relies on \cd{id/6} (setting the replica factor from 1 to \cd{RMax}) to determine a total eligible \cd{Placement}, possibly starting from a partial eligible placement \cd{PPlacement}. It then computes its \cd{Cost} by relying on predicate\footnote{The Prolog built-in predicate \cd{findall(X, Goal, Result)} collects in the list \cd{Result} all the items \cd{X} satisfying the given \cd{Goal}. The Prolog built-in predicate \cd{sumlist(List, Sum)} computes the \cd{Sum} of all elements in the given \cd{List}.} \cd{cost/2} (lines 39, 42--43) and its storage \cd{Allocation} via predicate \cd{allocatedStorage/2} (lines 40, 44--45). Storage allocation is simply represented as a list of pairs \cd{(N,Storage)} that keeps track of the \cd{Storage} used by each image at each node \cd{N}. Finally, predicate \cd{storePlacement/3} (lines 41, 46--49) stores the found \cd{Placement}, its storage \cd{Allocation}, and its \cd{Cost} in the knowledge base as a fact of the form
\medskip\begin{code}
placedImages(Placement, Allocation, Cost).
\end{code} 
\noindent that will later be used by our continuous reasoning approach. If a \cd{placedImages/3} fact already exists, it is retracted (line 47), and then the new one is asserted (line 48).

\medskip\begin{codeNum}[firstnumber=37, breaklines]  
placement(Images, Nodes, MaxR, PPlacement, Placement, Cost) :-
    id(Images, Nodes, PPlacement, Placement, 1, MaxR), 
    cost(Placement, Cost),
    allocatedStorage(Placement,Allocation),
    storePlacement(Placement, Allocation, Cost).

cost(Placement, Cost) :- 
    findall(CS, (member(at(I,N),Placement), image(I,S,_), node(N,_,C), CS is C*S), CSs), sum_list(CSs, Cost).

allocatedStorage(P, Alloc) :- 
    findall((N,S), (member(at(I,N), P), image(I,S,_)), Alloc).

storePlacement(Placement, Alloc, Cost) :-
    ( placedImages/3 -> retract(placedImages/3) ; true),
    assert(placedImages(Placement, Alloc, Cost)).
\end{codeNum}

\medskip\begin{example}
Querying \cd{id/6} over the knowledge base of Example~\ref{ex:encoding} returns the same placement \cd{P} of Example~\ref{ex:imagePlacement}, along with its associated cost of \texteuro 437. 

\medskip\begin{code}[breaklines]
?- placement([alpine,ubuntu,nginx], [cloud,edge1,edge2,edge3,edge4,edge5], 3, [], P, 1, Cost).
P = [at(nginx, edge5), at(nginx, edge3), at(nginx, cloud), at(ubuntu, edge5), at(ubuntu, edge1), at(ubuntu, cloud), at(alpine, cloud)].
Cost = 437.0 
\end{code}

\noindent 
As a side effect, it also asserts a fact \cd{placedImages(P,Alloc,Cost)} within the knowledge base. Such fact can be retrieved as follows

\medskip\begin{code}[breaklines]
? - placedImages(P, Alloc, Cost).
P = [at(nginx, edge5), at(nginx, edge3), at(nginx, cloud), at(ubuntu, edge5), at(ubuntu, edge1), at(ubuntu, cloud), at(alpine, cloud)],
Alloc = [(edge5, 192), (edge3, 192), (cloud, 192), (edge5, 69), (edge1, 69), (cloud, 69), (cloud, 8)],
Cost = 437.0.
\end{code}

\noindent Such asserted facts will be used in the next section for adapting image placements in response to infrastructure changes through continuous reasoning. \hfill $\qed$

\end{example}

%\damiano{codice rimasto uguale?}

\subsection{Solving CIPP with Prolog}
\label{sec:continuousreasoning}

Based on the above, we now describe a solution to CIPP (Problem~\ref{problem:continuous_image_placement}).
Predicate \cd{cr/7} (lines 49--58) inputs a list \cd{[I|Is]} of images to be placed, a list of target \cd{Nodes}, and the current image placement \cd{P} with its associated allocation \cd{Alloc}. By recurring on the list of images to be placed, it determines a portion \cd{NewPOk} of \cd{P}, which constitutes a partial eligible placement of \cd{[I|Is]} over \cd{Nodes}, and collects in the \cd{KO} list all images that need to be migrated. 
The first clause of  \cd{cr/7} retrieves the placement of image \cd{I} from the current placement \cd{P} and appends it to the partial eligible placement \cd{POk} being build by creating \cd{TmpPOk} (lines 50). Then, \cd{cr/7} checks\footnote{Note that \cd{cr/7} exploits an extended version of predicate \cd{storageOk/3}, which checks all placement nodes and inputs the current allocation \cd{Alloc}. This is needed to account for changes in the size of an image by unloading the previous allocation and summing up possibly updated image requirements. Interested readers can find the code of \cd{storageOk/4} in the \proimg{} repository. } whether for \cd{TmpPOk} Equation~\ref{eq:maxreplicas}, Equation~\ref{eq:transfertimes} and Equation~\ref{eq:hardwarelimit} still hold (lines 51--53). In such a case, \cd{TmpPOk} is passed to the recursive call, moving to the next image. 
Whenever such conditions do not hold true\footnote{Pursuing efficiency, instead of negating \cd{transferTimesOk/3} and \cd{storageOk/4} in the second clause of \cd{cr/7}, we use the cut operator (line 53) to prevent backtracking, thus avoiding that an image already included in \cd{POk} (i.e., matching the first clause of \cd{cr/7}) is then also included in \cd{KO} (by a later match with the third clause of \cd{cr/7}).}, the second clause of  \cd{cr/7} adds image \cd{I} to the list of \cd{KO} images (line 55) and recurs to the next image without extending the partial eligible placement being build (line 56). It is worth noting that this case also handles new images to be placed, which can be continuously added to the image repository, by including them in the \cd{KO} list. The recursion ends when the list of images to be checked is empty (line 57).

\medskip\begin{codeNum}[firstnumber=49, breaklines]  
cr([I|Is], Nodes, MaxR, P, POk, NewPOk, Alloc, KO) :-
    findall(at(I,N), member(at(I,N), P), INs), append(INs, POk, TmpPOk),
    length(INs, IReplicas), IReplicas =< MaxR,
    transferTimesOk(I, Nodes, TmpPOk),
    image(I, Size, _), storageOk(I, Size, TmpPOk, Alloc), !,
    cr(Is, Nodes, MaxR, P, TmpPOk, NewPOk, Alloc, KO).
cr([I|Is], Nodes, MaxR, P, POk, NewPOk, Alloc, [I|KO]) :-
    cr(Is, Nodes, MaxR, P, POk, NewPOk, Alloc, KO).
cr([], _, _, _, POk, POk, _, []).
\end{codeNum}

Predicate \cd{crPlacement/5} retrieves the current \cd{Placement} and its allocation \cd{Alloc} (line 60) and exploits \cd{cr/7} to identify a partial eligible \cd{PPlacement} and the list of \cd{KOImages} within \cd{Placement} (line 60). Finally, it relies on \cd{placement/6} to determine a complete eligible \cd{NewPlacement} that extends the identified partial one (line 61). If no previous placement exists or it is not possible to fix the current one via continuous reasoning, predicate \cd{crPlacement/5} attempts to determine a complete eligible \cd{InitPlacement} from scratch (line 62--63).

\medskip\begin{codeNum}[firstnumber=58, breaklines]
crPlacement(Images, Nodes, MaxR, NewPlacement, Cost) :- 
    placedImages(Placement, Alloc, _),
    cr(Images, Nodes, MaxR, Placement, [], OkPlacement, Alloc, KOImages), dif(KOImages, Images),
    placement(KOImages, Nodes, MaxR, OkPlacement, NewPlacement, Cost).
crPlacement(Images, Nodes, MaxR, InitPlacement, Cost) :- 
    placement(Images, Nodes, MaxR, [], InitPlacement, Cost).
\end{codeNum}

Lastly, predicate \cd{declace/2} (lines 64--68) can be called periodically over a changing knowledge base to continuously determine needed updates of the image \cd{Placement} and its associated \cd{Cost}. It simply retrieves input data for \cd{crPlacement/5} via predicates \cd{imagesToPlace/1}, \cd{networkNodes/1} and via the fact \cd{maxReplicas/1} (lines 65--67), and last it calls \cd{crPlacement/5} stopping at the first result through \cd{once/1} (line 68). 

Note that predicates \cd{imagesToPlace/1} and \cd{networkNodes/1} (lines 65 and 66, respectively) sort their output to drive backtracking via a combination of a \textit{fail-first} and a \textit{fail-last} heuristics~\cite{AIMA}. Particularly, \cd{Images} are sorted (and will be explored) by size in descending order so as to try to place larger images first, in that they might be supported by fewer candidate nodes (\textit{fail-first}).
Nodes are instead sorted by increasing cost and decreasing average outgoing bandwidth and hardware capabilities, thus pursuing the two-fold objective of containing deployment costs and avoiding filling up less capable nodes too early in the search (\textit{fail-last}). 

\medskip\begin{codeNum}
declace(Placement, Cost) :-
    imagesToPlace(Images), 
    networkNodes(Nodes), 
    maxReplicas(Max),
    crPlacement(Images, Nodes, Max, Placement, Cost), !.
\end{codeNum}

The above 68 lines of code, which already constitute the core of \proimg{}, heuristically solve the \textit{continuous image placement problem} (viz., Problem~\ref{problem:continuous_image_placement}) stated in Section~\ref{sec:problemstatement} by determining a total eligible placement of a set of images onto a target infrastructure, starting from a previous partial eligible placement. 
Note that the proposed solution adapts the current placement in response to changes in the image repository to be distributed (i.e., image addition/removal and storage requirements changes), in the target infrastructure (i.e., node/link crash, link QoS degradation), and in the maximum replica factor.

\medskip\begin{example}
Querying predicate \cd{declace/2} over the knowledge base of Example~\ref{ex:encoding} gives the following result

\medskip\begin{code}[breaklines]
?- declace(P, Cost).
P = [at(alpine, edge2), at(ubuntu, edge5), at(ubuntu, edge2), at(nginx, edge3), at(nginx, edge5), at(nginx, edge2)],
B = 308.0.
\end{code}

\noindent which corresponds to the eligible placement of Example \ref{ex:eligibleplacement}, costs \texteuro 308, and replicates \cd{alpine} onto one node, \cd{ubuntu} onto two nodes, and \cd{nginx} onto three nodes.

Assume that node \cd{edge3} fails. Querying again predicate \cd{declace/2} over the updated knowledge base adapts the previous placement as follows

\medskip\begin{code}[breaklines,commandchars=\\\{\}]
?- declace(PNew, Cost).
PNew = [at(nginx, edge5), at(nginx, edge2), at(alpine, edge2), at(ubuntu, edge5), at(ubuntu, edge2)],
Cost = 212.0.
\end{code}
\noindent which simply removes the \cd{nginx} image deployed at \cd{edge3}.
Similarly, adding a new image to place \cd{image(redis,149,60)} to the repository and querying \cd{declace/2} returns the following placement

\medskip\begin{code}[breaklines,commandchars=\\\{\}]
?- declace(PNew, Cost).
PNew = [\textcolor{OliveGreen}{at(redis, n5), at(redis, n2)}, at(alpine, n2), at(ubuntu, n5), at(ubuntu, n2), at(nginx, n5), at(nginx, n2)],
Cost = 331.2.
\end{code}

\noindent without affecting the previously determined one but only adding the placement for the \cd{redis} image, highlighted in \textcolor{OliveGreen}{green}. \hfill$\qed$

\end{example}

\rev{The described Prolog implementation can be easily extended with a predicate to solve OIPP. However, proving a placement \cd{P} optimal would require exploring the whole search space, incurring in exp-time complexity. This approach becomes practically unfeasible even over small problem instances and it is the reason why, in the next section, we present an ASP encoding for solving OIPP. As aforementioned, ASP is particularly well-suited to efficiently solve optimisation problems~\cite{DBLP:books/sp/Lifschitz19}.}

\subsection{Solving OIPP with ASP}
\label{subsec:asp_representation}

\rev{We now introduce an ASP encoding (according to the \textsc{GrinGo} language~\cite{potasscoGuide}) to solve the OIPP problem (Problem~\ref{problem:optimal_image_placement}). 
We follow the \emph{Guess, Check \& Optimise} paradigm, thus defining programs that (\textit{i}) generate the search space of possible image placements, (\textit{ii}) prune invalid ones, and (\textit{iii}) minimise an objective function to rank valid solutions.}

\begin{figure}
\begin{codeNum}[firstnumber=75]
1 { placement(I, N): node(N, _, _)} R :- image(I, _, _), maxReplicas(R).   

transfer_time(Img, Src, Src, 0) :- 
    placement(Img, Src).
transfer_time(Img, Src, Dst, @compute_transfer_time(S,B,L)) :- 
    placement(Img, Src),
    image(Img, S, _),
    link(Src, Dst, L, B).

transfer_ok(Img, Dst) :-
    transfer_time(Img, X, Dst, T),
    image(Img, _, MaxTransferTime),
    T <= MaxTransferTime.
    
:- image(Img,_,_), node(Dst, _, _), not transfer_ok(Img, Dst).    

:- TS = #sum{Size, Img : image(Img, Size, _), 
        placement(Img, N)}, node(N, Cap, _), TS > Cap. 

:~ placement(Img, X), node(X, _, Cost), 
   image(Img, Size, _). [Cost * Size, Img, X]  
\end{codeNum}
    \caption{ASP encoding of OIPP.}
    \label{fig:ASPsolution}
\end{figure}

\rev{The ASP input follows the format introduced in Section~\ref{sec:kr}. Figure~\ref{fig:ASPsolution} lists the complete ASP encoding.
The first requirement is for each image to be deployed on at least one (and at most \cd{R}) nodes as per Equation~\ref{eq:maxreplicas} of Definition~\ref{def:eligiblePlacement}. This is succintly encoded at line 75, where the \cd{placement(I,N)} atom models that each image \cd{I} is placed onto a non-empty set of at most \cd{R} nodes \cd{N}.}

\rev{The rules at lines 76--81 account for the computation of transfer times, where \cd{transfer_time(I, Src, Dst, T)} models that \emph{it is possible to transfer image \cd{I} from node \cd{Src} to node \cd{Dst} in \cd{T} time units}. The first rule (line 76--77) states that if the image \cd{Img} is placed on node \cd{Src} its transfer time towards \cd{Src} itself is null.
The second rule computes transfer times in the general case as in the Prolog code (lines 19--24) and relies on the \atterm \cd{compute\_transfer\_time}, which executes a Python function to compute transfer times of images between pairs of nodes -- since ASP does not natively support floating point operations.}
Last, the rule \cd{transfer\_ok(Img,Dst)} models that the image \cd{Img} can be pulled from node \cd{Dst} within its maximum transfer time.

\rev{We now define \emph{constraints} to prune invalid image placements generated by the above choice rule.
First, each image's transfer time must meet the specified maximum latency requirements as per Equation~\ref{eq:transfertimes} of Definition~\ref{def:eligiblePlacement}. Thus, our ASP enconding discards placements that do no allow a node to pull an image within its maximum transfer time (line 86).}
\rev{Furthermore, the images placed on the same node should not exceed its capacity as per Equation~\ref{eq:hardwarelimit} of Definition~\ref{def:eligiblePlacement}. This is achieved by discarding solutions that do not meet this constraint (line 87--88).
Here, \cd{TS} is the sum (computed with the aggregate \cd{\#sum}) of the sizes \cd{Size} over all the pairs \cd{(Size, Img)} for an image \cd{Img} with size \cd{Size} (line 87) placed on the node \cd{N} with maximum capacity \cd{Cap} (line 88). The above lines determine all solutions to IPP. } 

\rev{Finally, we aim at finding a cost-optimal image placement, i.e., one that minimise storage costs as per Definition~\ref{def:costOptimal}. We can do so by adding the \emph{weak constraint} at line 89--90.
Weak constraints enable defining a \emph{preference criterion} over a logic program's answer sets, akin to an \emph{objective function} in an optimisation problem. The weak constraint in our program expresses the following: when there is a placement of an image \cd{Img} of size \cd{Size} on a node \cd{X} that has monthly a cost per MB of \cd{Cost}, such placement contributes to the cost of the answer set as the additive term \cd{Cost * Size} once for each distinct pair of the tuple \cd{(Img,X)}.
Optimal answer sets are those with minimal costs.}

\subsection{Functioning of \proimg{} }%\damiano{proimg architecture? titolo alternativo}}
\label{sec:architecture}

We assume that \proimg{} periodically checks the state of the image placement exploiting  Prolog continuous reasoning in combination with ASP-based optimisation. Figure~\ref{fig:flow} sketches the overall functioning of \proimg{}. 
ASP is used to solve OIPP when (\textit{i})~an \textit{initial image placement} needs to be determined, or (\textit{ii}) the heuristic \textit{continuous reasoning}, invoked via \cd{declace/2}, \textit{fails to fix the current placement} after an infrastructure or repository change. Informally, the ASP encoding acts as a backup for the second clause of \cd{crPlacement/5} of Section~\ref{sec:continuousreasoning}. 

\input{img/declace_workflow}

Such behaviour, which we will thoroughly assess in the next section, suitably alternates ASP optimisation and continuous reasoning adaptation steps, balancing between prompt adaptation and cost optimisation of found placements. On one hand, it ensures optimal initial placements and optimal re-placements when the current status cannot be recovered. On the other hand, it features faster decision-making in the adaptation phase at runtime, by focussing only on those images that need to be migrated to reduce service disruption and avoid unnecessary image transfers.
Indeed, ASP solves OIPP exactly with the disadvantage of possibly incurring longer execution times over small problem instances than our Prolog heuristic, which solves problem instances reduced via continuous reasoning, taming their NP-hard complexity.

At the time of initial placement, we may want to find an optimal solution, since images are yet to be deployed and we could wait some more time in favour of cost minimisation. Once deployed, instead, the goal is to react as fast as possible to network failures or changes in the image requirements.
Hence, quickly determining a (sub-optimal) eligible placement and reducing the number of images to be migrated through continuous reasoning may be preferred to finding a cost-optimal solution. 
Analogously, ASP is used when continuous reasoning fails to adapt the current placement. This situation reflects some critical network scenarios that might require some more attention for determining a new placement, thus tolerating longer execution times.

%% file: img/declace_workflow.tex
\begin{figure}[ht]
    \centering
    \resizebox{0.95\textwidth}{!}{%
    \begin{tikzpicture}[node distance=2cm]

    \node (start) [start] {};
    \node (solveOIPPASP) [process, below of=start, yshift=0.5cm] {solve OIPP with ASP};
    \node (dummy_solve) [right of=solveOIPPASP,xshift=0.5cm] {};
    
    \node (found) [decision, below of=solveOIPPASP, yshift=-1cm] {
    found
    placement?
    };
    
    \node (waitNo) [process, left of=found, xshift=-2cm] {wait\\ monitoring period};
    
    \node (waitYes) [process, right of=found, xshift=2cm] {wait\\ monitoring period};
    \node (waitYes_dummy) [below of=waitYes] {};
    
    \node (CIPPProlog) [process, right of=waitYes, xshift=2cm] {solve CIPP with Prolog};
    
    \node (foundPL) [decision, right of=CIPPProlog,xshift=2cm] {found placement?};
    \node (dummy_found_bottom_yes) [below of=foundPL] {};

    \draw [arrow] (start) -- (solveOIPPASP);
    \draw [arrow] (solveOIPPASP) -- (found);
    \draw [arrow] (found) -- node[anchor=north] {yes} (waitYes);
    \draw [arrow] (found) -- node[anchor=north] {no} (waitNo);
    \draw [arrow] (waitNo) |- (solveOIPPASP);
    \draw [arrow] (waitYes) -- (CIPPProlog);
    \draw [arrow] (CIPPProlog) -- (foundPL);
    
    \draw [arrow] (foundPL) |- node[anchor=south east] {no} (solveOIPPASP);

    \draw [] (foundPL) |- node[anchor=north west] {yes} (dummy_found_bottom_yes);
    \draw [] (dummy_found_bottom_yes) -| (waitYes_dummy);
    \draw [] (dummy_found_bottom_yes) |- (waitYes_dummy);
    \draw [arrow] (waitYes_dummy) -| (waitYes);

    \end{tikzpicture}
    }
    \caption{Workflow of \proimg{}.}
    \label{fig:flow}
\end{figure}

% OLD
% \begin{figure}[ht]
%     \centering
%     \resizebox{0.4\textwidth}{!}{%
%     \begin{tikzpicture}[node distance=2cm]
%     \node (start) [start] {};
%     \node (solveOIPPASP) [process, below of=start, yshift=0.5cm] {solve OIPP with ASP};
%     \node (found) [decision, below of=solveOIPPASP, yshift=-1cm] {
%     found
%     placement?
%     };
    
%     \node (waitYes) [process, below of=found, yshift=-1cm] {wait\\ monitoring period};
%     \node (waitNo) [process, left of=found, xshift=-2cm] {wait\\ monitoring period};
%     \node (dummy_dec) [right of=found,xshift=0.5cm] {};
%     \node (CIPPProlog) [process, below of=waitYes] {solve CIPP with Prolog};
%     \node (dummy_cipp) [right of=CIPPProlog,xshift=0.5cm] {};

%     \draw [arrow] (start) -- (solveOIPPASP);
%     \draw [arrow] (solveOIPPASP) -- (found);
%     \draw [arrow] (found) -- node[anchor=east] {yes} (waitYes);
%     \draw [arrow] (found) -- node[anchor=north] {no} (waitNo);
%     \draw [arrow] (waitNo) |- (solveOIPPASP);
%     \draw [arrow] (waitYes) -- (CIPPProlog);
%     \draw [thick] (CIPPProlog) -| (dummy_cipp);
%     \draw [thick] (dummy_cipp) -| (dummy_dec);
%     \draw [thick] (dummy_cipp) |- (dummy_dec);
%     \draw [arrow] (dummy_dec) |- (found);
    
%     \end{tikzpicture}
%     }
%     % \includegraphics[trim={1cm 1cm 1cm 1cm},clip,width=0.5\textwidth]{img/flow}
%     \caption{Workflow of \proimg{}.}
%     \label{fig:flow}
% \end{figure}

%% file: src/experiments.tex
\section{Experimental Assessment}\label{sec:experiments}

\begin{table}[t]
\centering
\caption{Experimental image set.}
\label{tab:images}
{\footnotesize
\begin{tabular}{|c|c|c|c| |c|c|c|c|}
\hline
\textbf{Id} & \textbf{Image name} & \textbf{Size {[}MB{]}} & \textbf{$T_{max}$} & \textbf{Id} & \textbf{Image name} & \textbf{Size {[}MB{]}} & \textbf{$T_{max}$} \\ \hline
1  & \cd{busybox}    & 4    & 15    & 7  & \cd{httpd}      & 195  & 60    \\ \hline 
2  & \cd{memcached}  & 126  & 30    & 8  & \cd{postgres}   & 438  & 120   \\ \hline
3  & \cd{nginx}      & 192  & 60    & 9  & \cd{ubuntu}     & 69   & 15    \\ \hline
4  & \cd{mariadb}    & 387  & 120   & 10 & \cd{redis}      & 149  & 30    \\ \hline
5  & \cd{alpine}     & 8    & 15    & 11 & \cd{rabbitmq}   & 201  & 60    \\ \hline
6  & \cd{traefik}    & 148  & 30    & 12 & \cd{mysql}      & 621  & 120   \\ \hline

\end{tabular}}
\end{table}

\noindent
In this section, we report and discuss the results of our experiments\footnote{All experiments run on a machine featuring Ubuntu 20.04.5 LTS equipped with 64GB of RAM and a 2.3GHz Intel Xeon Gold 5218 16-core processor.}  with the prototypes illustrated in Section~\ref{sec:methodology}. In the following, we will assess, compare and contrast via simulation the performance of both the individual components of \proimg and of their combination as per the workflow depicted in Figure~\ref{fig:flow}. %After detailing the experimental setup and configurations (Section~\ref{sec:experimental_setup}), we illustrate the obtained results (Section~\ref{sec:experimental_results})  and, finally, we sum up some lessons learnt (Section~\ref{sec:lessons}). 

\subsection{Assessment of ASP and Iterative Deepening}
\label{sec:aspvsprolog}

In this first experiment, we compare the performance of our iterative deepening solution and ASP solutions to IPP and OIPP, respectively. Particularly, we consider the execution times of both solutions and the cost of their output placements. 

\medskip\noindent\textbf{\textit{Setup}} We consider the set of images of Table~\ref{tab:images}, and we group them into three subsets, each containing four images: $I_1=\{\cd{busybox}, $ $ \cd{memcached},$ $ \cd{nginx}, $ $ \cd{mariadb}\}$,  $ I_2=\{\cd{alpine}, \cd{traefik},$ $ \cd{httpd},$ $\cd{postgres}\}$, and $I_3=\{\cd{ubuntu},$ $ \cd{redis},$ $\cd{rabbitmq},$ $ \cd{mysql}\}$. Note that the three subsets are homogeneous in that they contain one small image (with maximum transfer time $T_{max}=15$ seconds), two medium images ($T_{max}=30$ and $T_{max}=60$ seconds, respectively), and one large image ($T_{max}=120$ seconds).

We then generate random infrastructures at varying number of nodes $n$. Each infrastructure is generated according to a Barabasi-Albert topology~\cite{barabasi1999emergence}, with preferential attachment parameter $m=3$ by relying on the Python's NetworkX library~\cite{hagberg2008networkx}. The Barabasi-Albert %algorithm builds a graph of $n$ nodes sequentially, by adding one node at a time. The $i$-th node is linked to a sample of $m$ previously added nodes. The sampling process is weighted on the degree of the nodes, so that high-degree nodes are sampled with higher probability. Thus, high-degree nodes are more likely to accrue more neighbours. This 
yields graphs with inhomogeneous degree distributions that are often observed in real-life networks~\cite{barabasi1999emergence}, characterised by \textit{hubs} with a high degree (e.g. Cloud nodes), while most of the nodes have a lower degree (e.g. edge noes), such as computer networks and the Internet~\cite{BARABASIVABENEPERLERETI}. %In our context, the high-degree nodes in the generated networks could be understood as cloud nodes, that intuitively are more central in the network compared to edge nodes, which are natural candidates to host local image repositories. 
As the first nodes in the network are likely to be connected to many other nodes in Barabasi-Albert topologies, we generate infrastructures with three extra nodes and we then remove the first three nodes from the obtained graph to avoid that they can host all image replicas.

On top of the generated network topology, each node gets assigned its storage according to a uniform discrete distribution: 10\% of the nodes get 2000 MB, 50\% of the nodes get 4000 MB and, 40\% of the nodes get either 8000 MB or 16000 MB.
Links in the generated network are assigned a latency in between 1 and 10 ms, and a bandwidth in between 25 and 1000 Mbps. Last, end-to-end links are computed by determining shortest paths between nodes based on latency through Dijkstra algorithm, and constraining them to the minimum bandwidth value along the path. %This reasonably mimics the variety of hardware and network resources that characterises cloud-edge settings.

This first experiment consists of trying to place from scratch the set of images $I_1$, $I_1 \cup I_2$ and $I_1 \cup I_2 \cup I_3$ onto 1000 randomly generated infrastructures per each infrastructure size $n \in \{ 25, 50, 75,$ $ 100, 125, 150\}$. We set the maximum amount of replicas $R=15$. For all runs, we measure execution times (with a timeout of $45$ seconds) and output placement costs.

\medskip\noindent\textbf{\textit{Results}} Figure~\ref{fig:4images},~\ref{fig:8images}, and~\ref{fig:12images} show the average execution times and placement costs obtained over 1000 random instances when placing $I_1$ (4 images), $I_1 \cup I_2$ (8 images) and $I_1 \cup I_2 \cup I_3$ (12 images) onto varying infrastructure sizes. We only consider instances for which both ASP and iterative deepening can succeed, i.e., we do not account for problem instances that are labelled as \textit{unsatisfiable} by our ASP solution. %\textcolor{black}{In this experiment, we consider only instances for which a solution exists.}

Figure~\ref{fig:4images_time} shows that iterative deepening is faster than ASP in placing 4 images for infrastructure of at least 75 nodes. Particularly, it is at least 1--2 seconds faster than ASP for infrastructure sizes above 100 nodes; this corresponds to iterative deepening being between 2\% (25 nodes) and 43\% (150 nodes) faster than ASP in determining eligible placements. For what concerns cost, instead, iterative deepening determines solutions that are on average more expensive by a factor between 4\% (25 nodes) and 18\% (150 nodes) than the optimal solution identified by ASP (Figure~\ref{fig:4images_cost}). In our experiments, this constitutes a difference of at most 45 euro per month for 150 nodes. 

Similarly, Figure~\ref{fig:8images_time} shows that iterative deepening is faster than ASP in placing 8 images, already for infrastructures larger than 50 nodes. This corresponds to execution times that are between 11\% (50 nodes) and 58\% (150 nodes) faster, in favour of iterative deepening. Again, such a performance increase is paid at the price of solutions that are sub-optimal by a percentage between 7\% (25 nodes) and 21\% (150 nodes) (Figure~\ref{fig:8images_cost}). This turns out in execution times of at most 3 seconds for iterative deepening and of almost 8 seconds for ASP on average, with a cost difference of at most 115 euro per month on average. 

Then, Figure~\ref{fig:12images_time} shows that iterative deepening is always faster than ASP in placing 12 images. In this latter case, which represents more complex problem instances, iterative deepening is between 20\% (25 nodes) and 81\% (150 nodes) faster than ASP in determining solutions that are between 5\% (25 nodes) and 15\% (150 nodes) far from optimal (Figure~\ref{fig:12images_cost}). This corresponds to at most 5-second execution times for iterative deepening, with a cost increase of at most 150 euro per month on average, compared to ASP execution times that reach 22 seconds on average.

Note that, while ASP always determines eligible placements -- if any exists -- the heuristic approach incurs in some instances that cannot be solved within the timeout. Such number increases along with the size of the considered problem instances (from $0.1\%$ -- 8 images, 50 nodes -- up to $11.8\%$ of the instances -- 12 images, 150 nodes).

% \vspace{-0.6cm}
\begin{figure*}[htb]
\begin{subfigure}{.5\textwidth}
  \centering
  \includegraphics[width=\linewidth]{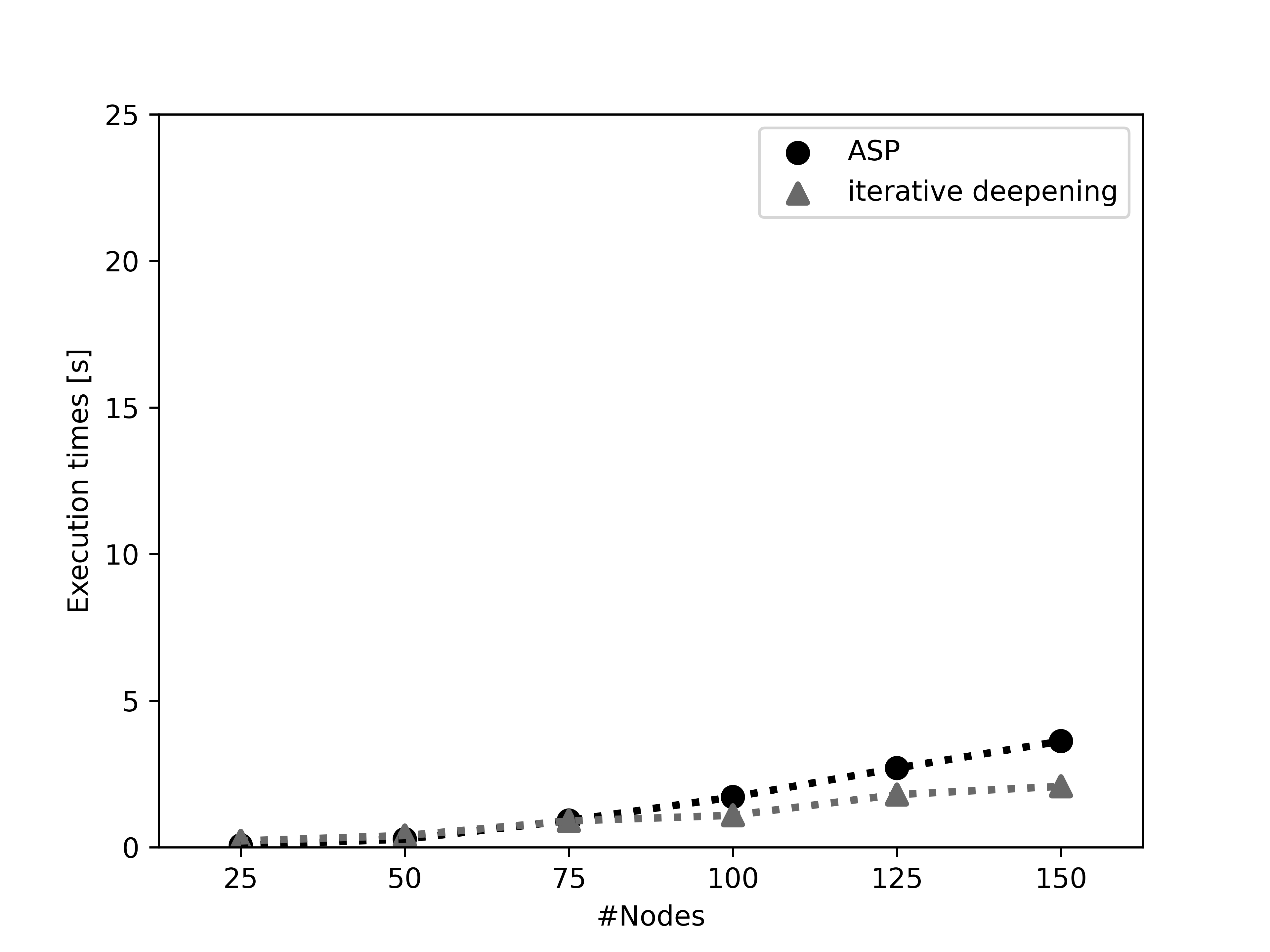}
  \caption{execution times}
  \label{fig:4images_time}
\end{subfigure}%
\begin{subfigure}{.5\textwidth}
  \centering
  \includegraphics[width=\linewidth]{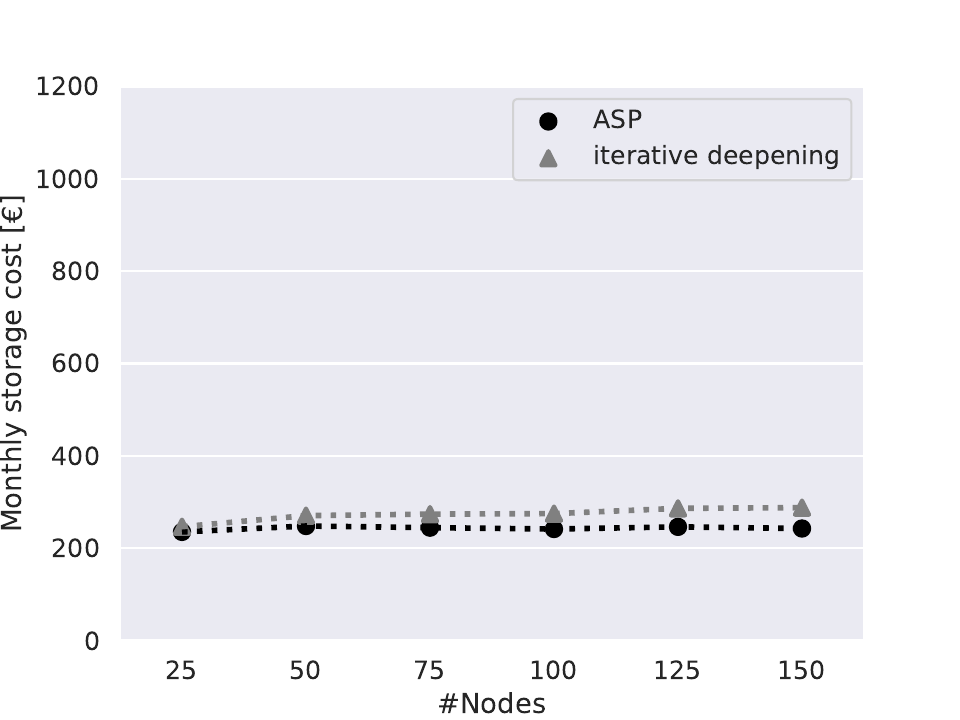}
  \caption{costs}
  \label{fig:4images_cost}
\end{subfigure}
\caption{ASP versus iterative deepening at varying number of nodes (4 images).}
\label{fig:4images}
\end{figure*}

\begin{figure*}[htb]
\begin{subfigure}{.5\textwidth}
  \centering
  \includegraphics[width=\linewidth]{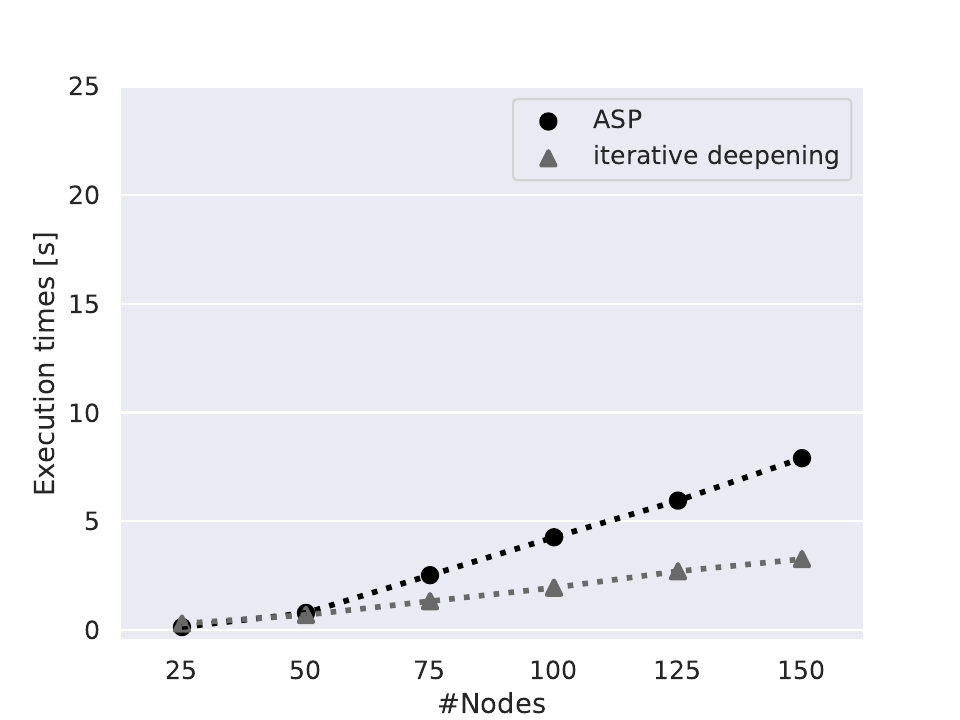}
  \caption{execution times}
  \label{fig:8images_time}
\end{subfigure}%
\begin{subfigure}{.5\textwidth}
  \centering
  \includegraphics[width=\linewidth]{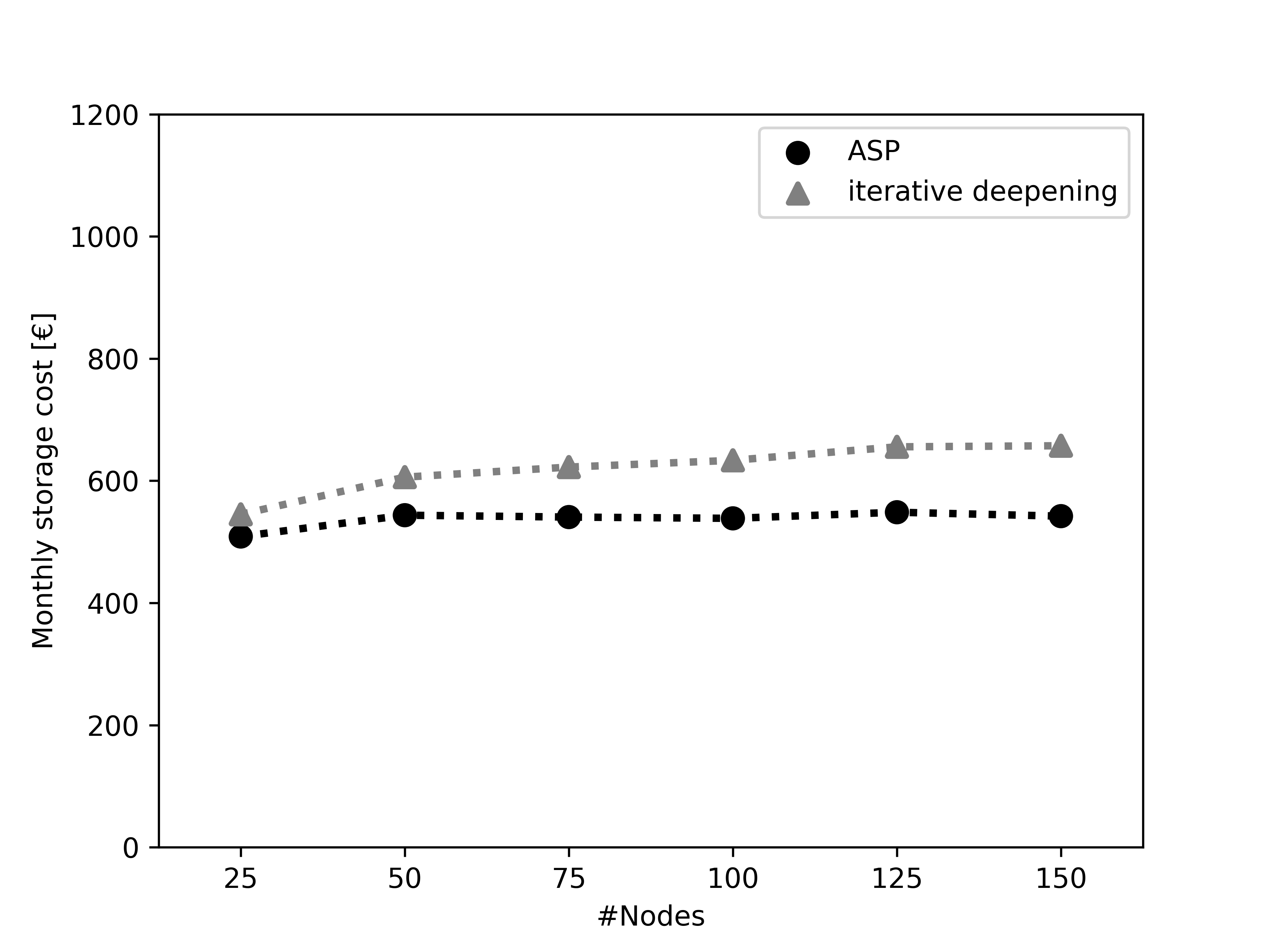}
  \caption{costs}
  \label{fig:8images_cost}
\end{subfigure}
\caption{ASP versus iterative deepening at varying number of nodes (8 images).}
\label{fig:8images}
\end{figure*}

\begin{figure*}[htb]
\vspace{-0.7cm}
\begin{subfigure}{.5\textwidth}
  \centering
  \includegraphics[width=\linewidth]{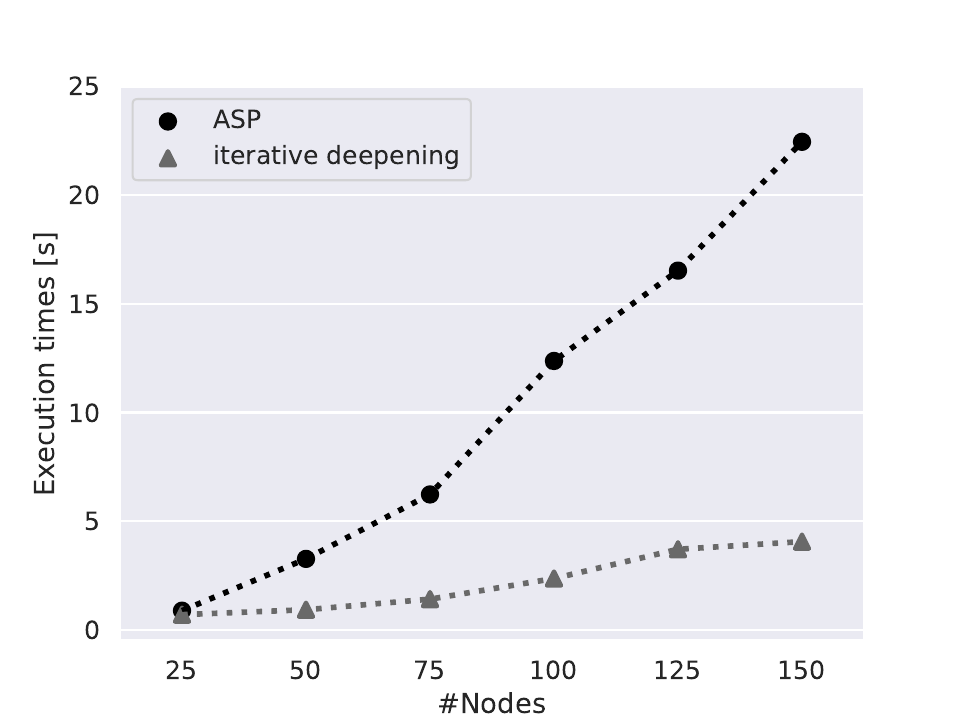}
  \caption{execution times}
  \label{fig:12images_time}
\end{subfigure}%
\begin{subfigure}{.5\textwidth}
  \centering
  \includegraphics[width=\linewidth]{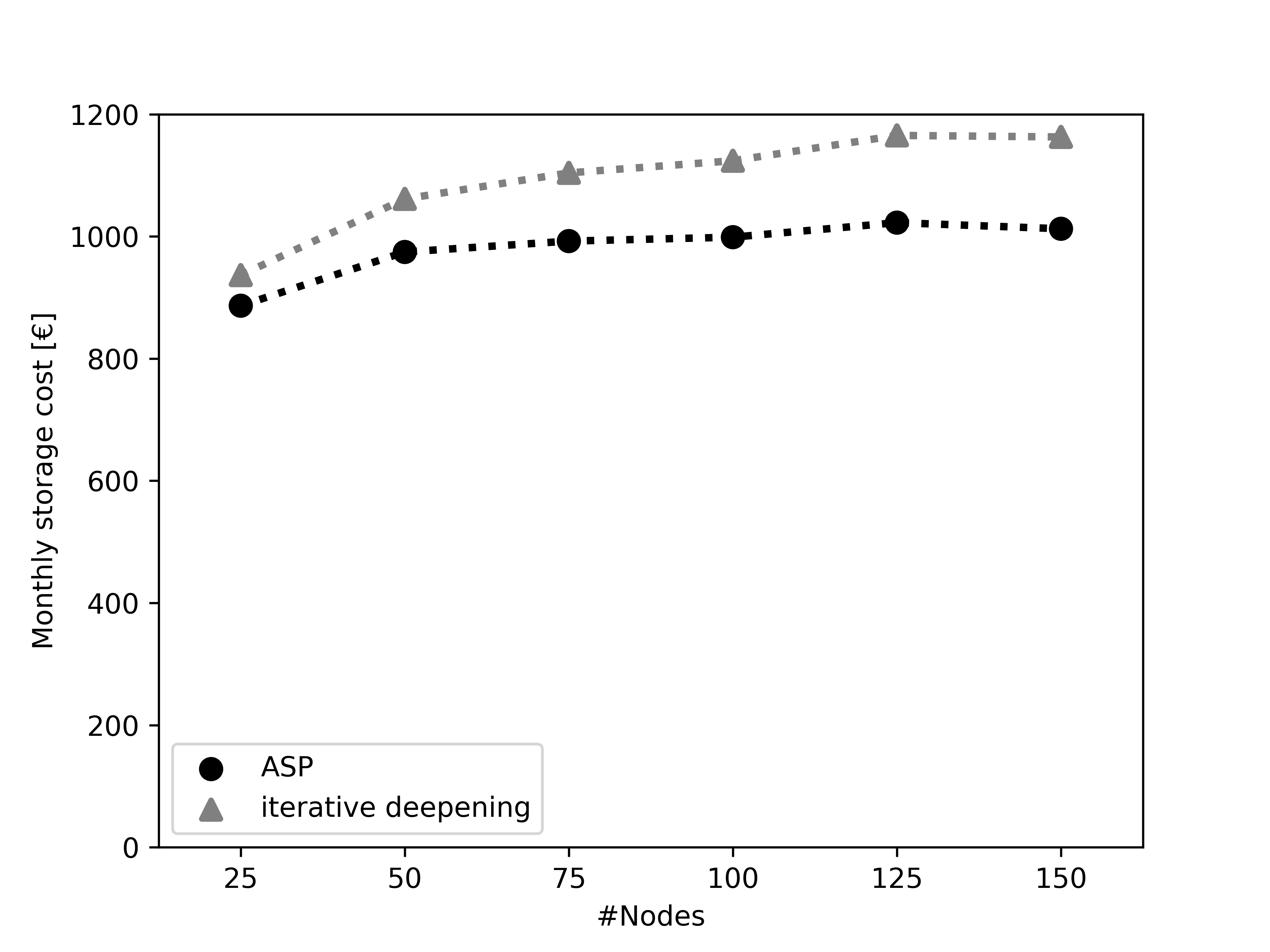}
  \caption{costs}
  \label{fig:12images_cost}
\end{subfigure}
\caption{ASP versus iterative deepening at varying number of nodes (12 images).}
\label{fig:12images}
\end{figure*}

\vspace{-0.9cm}
\medskip\noindent\textbf{\textit{Discussion}} These results confirm the analyses of Section~\ref{sec:architecture}:
\begin{itemize}
    \item \textit{iterative deepening is in general faster than ASP} (by a factor between 2\% and 81\%) in solving the considered container image placement problem both at varying infrastructure sizes and number of images to be placed,
    \item \textit{ASP always finds exact solutions to OIPP}, while heuristic iterative deepening determines solutions that are between 4\% and 18\% far from optimal, and shows an increasing failure rate while handling larger problem instances,
    \item \textit{iterative deepening stays closer to the optimal when placing fewer images} irrespectively of the size of the target infrastructure, as per Figures \ref{fig:4images} and \ref{fig:8images}.
\end{itemize}

Finally, the above experimental results allow us to estimate reasonable timeout thresholds for both our ASP encoding (which did not exceed 25 seconds on average) and heuristic iterative deepening solution (which did not exceed 6 seconds on average). %Such values will be used in the next experiment.

\subsection{Assessment of \proimg adaptive behaviour}
\label{sec:declace_assessment}

In this second experiment, we will assess the performance of \proimg in continuously adapting the placement of container images over cloud-edge resources, i.e.,~in solving CIPP over lifelike scenarios at an increasing number of nodes and images to place. As before, we focus on execution times and placement costs achieved by \proimg against an ASP-computed baseline. 

\medskip\noindent\textbf{\textit{{Setup}}} We consider again the set of images reported in \cref{tab:images}, and group them as before into the same subsets of images to place, viz. $I_1$, $I_1 \cup I_2$ and $I_1 \cup I_2 \cup I_3$.
For each amount of images to place, we generate a random Barabasi-Albert infrastructure again setting parameter $m=3$. Each node is initially assigned 4000 MB of storage, while end-to-end links are generated with the same QoS profiles and routing procedure as in the previous experiment. 
Differently from the previous experiment, since we are interested in evaluating our methods in an dynamic scenario, we do not generate multiple random infrastructures but simulate life-like network events by modifying an initial random infrastructure and set of images for $1000$ epochs. At each epoch, the following \textcolor{black}{changes in the infrastructure} can happen:

\begin{itemize}
    \item \textit{node failure}: each node in the network can fail (and be removed from the infrastructure for one\footnote{To avoid complete network degradation and considerable reduction of the network size, failed nodes are started anew in the subsequent simulation step. Since node failures affect network for a single simulation step, this implicitly models re-connection of nodes in the infrastructure.} epoch.) with a probability of $5\%$,
    \item \textit{link QoS and node storage variations}: with a probability of $0.5$ each node storage, and each link latency and bandwidth profile\footnote{To simulate network changes in a lifelike manner, latency increases are accompanied by bandwidth decreases (\textit{QoS degradation}) and, vice versa, latency decreases are accompanied by bandwidth increases (\textit{QoS improvement}).} change by a uniformly sampled factor of $\pm 15\%$, and
    \item \textit{image requirements variations}: with a probability of $0.1$, each image storage requirement change by a uniformly sampled factor of $\pm 5\%$.
\end{itemize}

\textcolor{black}{Link and node attributes are directly affected (depending on the type of change), while node-to-node metrics are re-computed as we sketched in the previous experiment and are indirectly affected by the simulated network changes. Thus, each simulation step yields a network infrastructure and image requirements that slightly differ from the previous configuration.}
%
%Following node failures, we recompute node-to-node bandwidths and latencies as in the previous experiment, computing the transitive closure of our graph with sum-of-latencies and minimum bandwidths among shortest length paths computed by means of Dijkstra algorithm.
%
\textcolor{black}{On each step of the simulation, we run \proimg (according to the workflow in Figure~\ref{fig:flow}) and the vanilla version of our ASP enconding to collect their respective execution times. Continuous reasoning Prolog steps are executed with 6-second timeouts, while ASP-based reasoning steps are executed with a timeout of 25 seconds (both in \proimg and in the vanilla version), as per our previous experiment.} %spostato da giù  

\medskip\noindent\textbf{\textit{Results}} \cref{fig:d4images}, \cref{fig:d8images}, and \cref{fig:d12images} show the results of the three simulations we run, aggregating average execution times and placement costs at varying infrastructure sizes and for 4, 8 and 12 images, respectively. 

As per \cref{fig:d4images}, \proimg and ASP show similar performances both in terms of execution times and obtained placement costs, with \proimg slightly worse on both metrics. Note that, despite this, the optimal ASP solution guarantees a cost-optimal placement but does not consider the implicit cost paid to discard (a possibly large part of) the current image placement (intuitively, the ``cost of moving images around''). On the other hand, continuous reasoning attempts to contain this cost by identifying and migrating only those images that do not currently satisfy their transfer time and storage requirements, in view of evolving infrastructure conditions.

On the other hand, \cref{fig:d8images} already shows the benefits of continuous reasoning in reducing decision-making times for infrastructures of at least 50 nodes. Indeed, as per \cref{fig:d8images_time}, \proimg shows an average consistent reduction in the range of 1.4--1.9 seconds, corresponding to a reduction in the execution times that ranges from  11\% (150 nodes) to 21\% (75 nodes). For what concerns cost, \cref{fig:d8images_cost} shows a slight increase in the range between 3\% (150 nodes) and 19.8\% (75 nodes).

Finally, \cref{fig:d12images} shows how more complex settings can benefit from continuous reasoning on any of the considered infrastructure sizes. Particularly, as per \cref{fig:d12images_time}, \proimg improves on execution times by a factor between 19.2\% (150 nodes) and 85\% (25 nodes). On the other hand, cost increases by a factor between 3.9\% (150 nodes) and 23.9\% (25 nodes).

\medskip\noindent\textbf{\textit{Discussion}} Summing up, we obtained the following experimental evidence:

\begin{itemize}
    \item \textit{execution times lower than ASP}: as for problem instances with at least 8 images and 75 nodes, \proimg always performs better than ASP by reducing execution times of a portion between $11\%$ and $85\%$ and is inversely proportional to the size of the target infrastructure (i.e., the larger the infrastructure, the lower the gain in terms of execution times),
    \item \textit{limited cost increase}: as for problem instances with at least 8 images and 50 nodes, the advantages in terms of execution times brought by \proimg correspond to a cost increase between 3\% and 24\%, also inversely proportional to the size of the target infrastructure, 
    \item \textit{avoidance of unnecessary migrations}: since continuous reasoning focuses on computing feasible placements by migrating only images that violate constraints in Definition~\ref{def:eligiblePlacement}, which turns out in a considerable reduction in terms of execution times when the current network conditions still support the previously enacted placement. 
\end{itemize}

\noindent Overall, the above results show that interleaving ASP with continuous reasoning based on iterative deepening achieves a positive trade-off between execution times and the cost of the found solution placement.

\begin{figure*}[htb]
\begin{subfigure}{.5\textwidth}
  \centering
  \includegraphics[width=\linewidth]{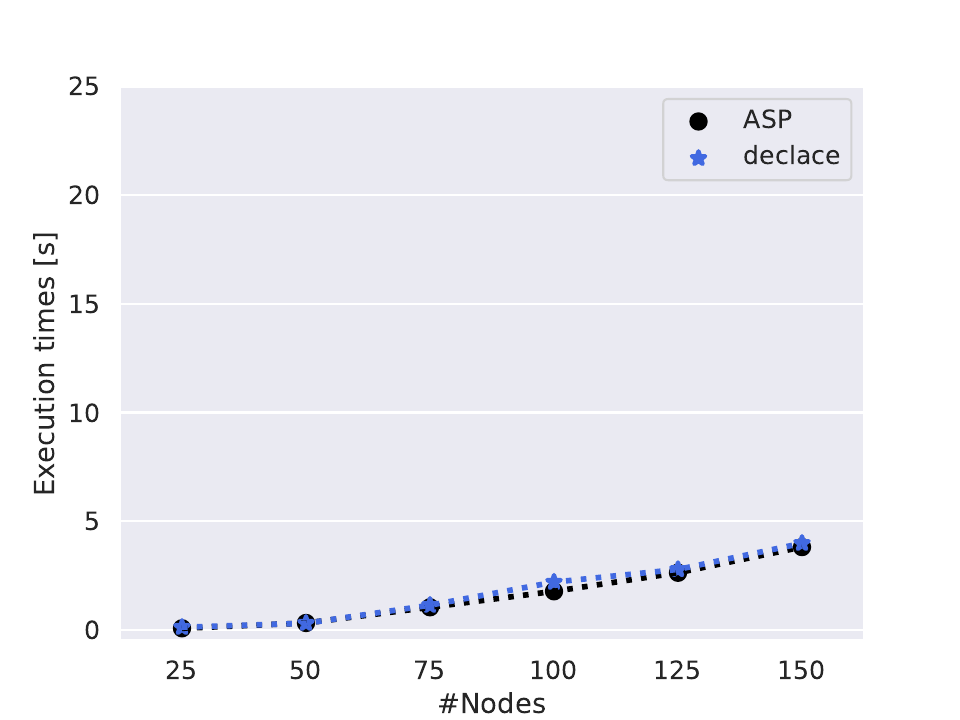}
  \caption{execution times}
  \label{fig:d4images_time}
\end{subfigure}%
\begin{subfigure}{.5\textwidth}
  \centering
  \includegraphics[width=\linewidth]{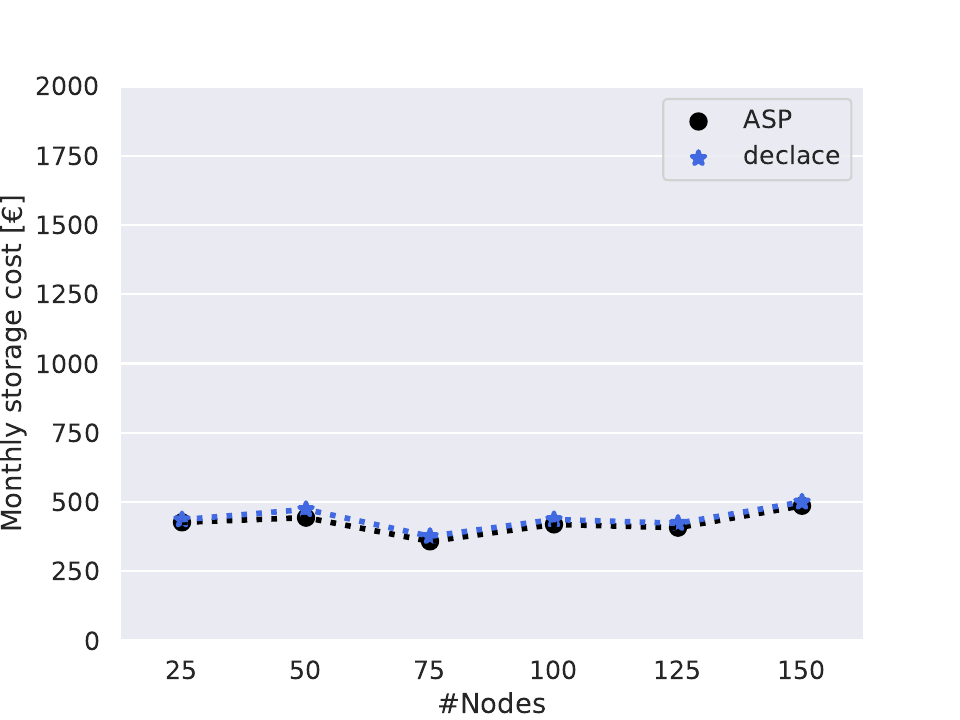}
  \caption{costs}
  \label{fig:d4images_cost}
\end{subfigure}
\caption{\proimg versus ASP at varying number of nodes (4 images).}
\label{fig:d4images}
\end{figure*}
\begin{figure*}[htb]
\begin{subfigure}{.5\textwidth}
  \centering
  \includegraphics[width=\linewidth]{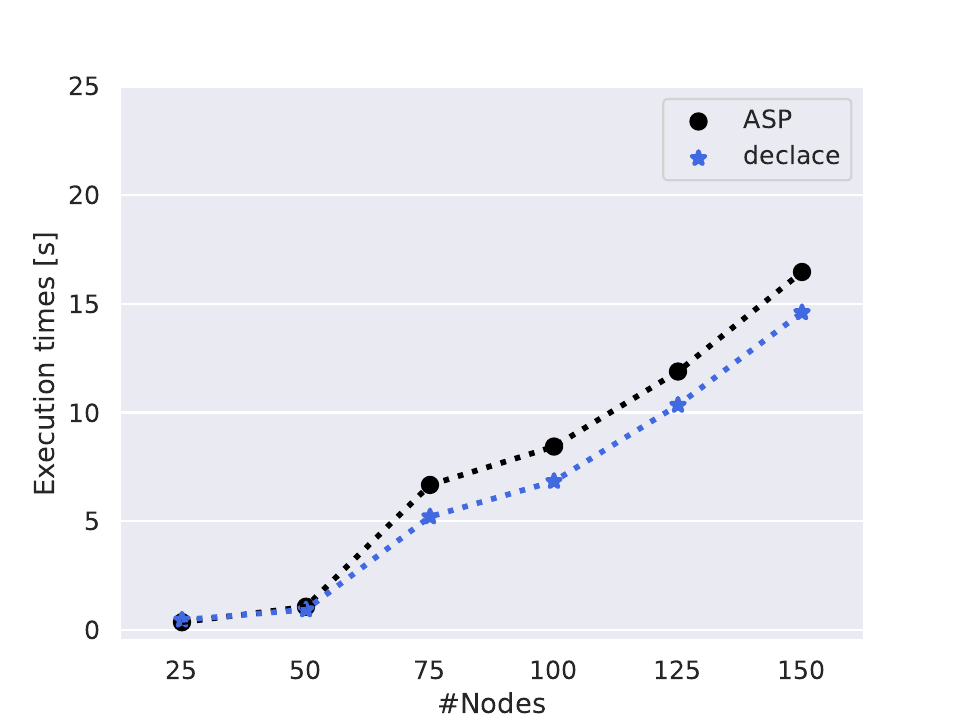}
  \caption{execution times}
  \label{fig:d8images_time}
\end{subfigure}%
\begin{subfigure}{.5\textwidth}
  \centering
  \includegraphics[width=\linewidth]{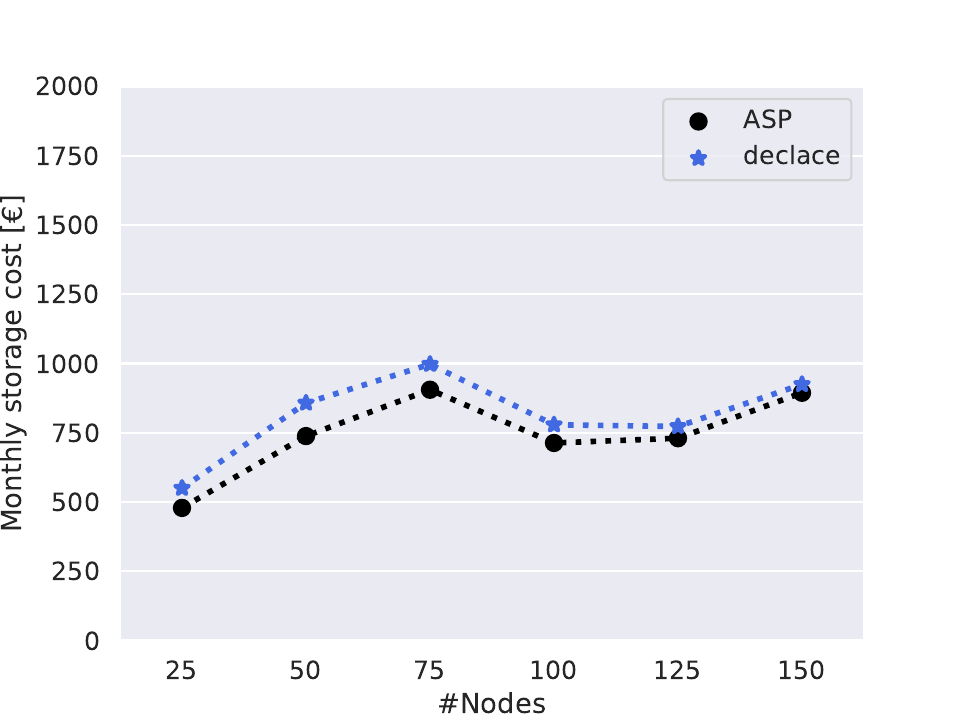}
  \caption{costs}
  \label{fig:d8images_cost}
\end{subfigure}
\caption{\proimg versus ASP at varying number of nodes (8 images).}
\label{fig:d8images}
\end{figure*}
\begin{figure*}[htb]
\begin{subfigure}{.5\textwidth}
  \centering
  \includegraphics[width=\linewidth]{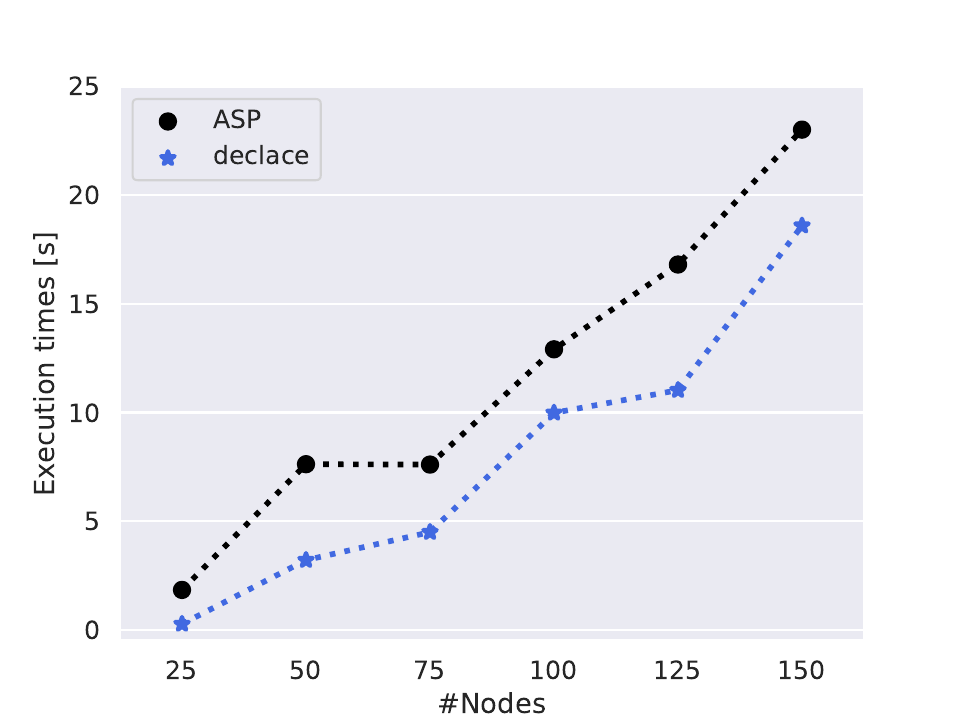}
  \caption{execution times}
  \label{fig:d12images_time}
\end{subfigure}%
\begin{subfigure}{.5\textwidth}
  \centering
  \includegraphics[width=\linewidth]{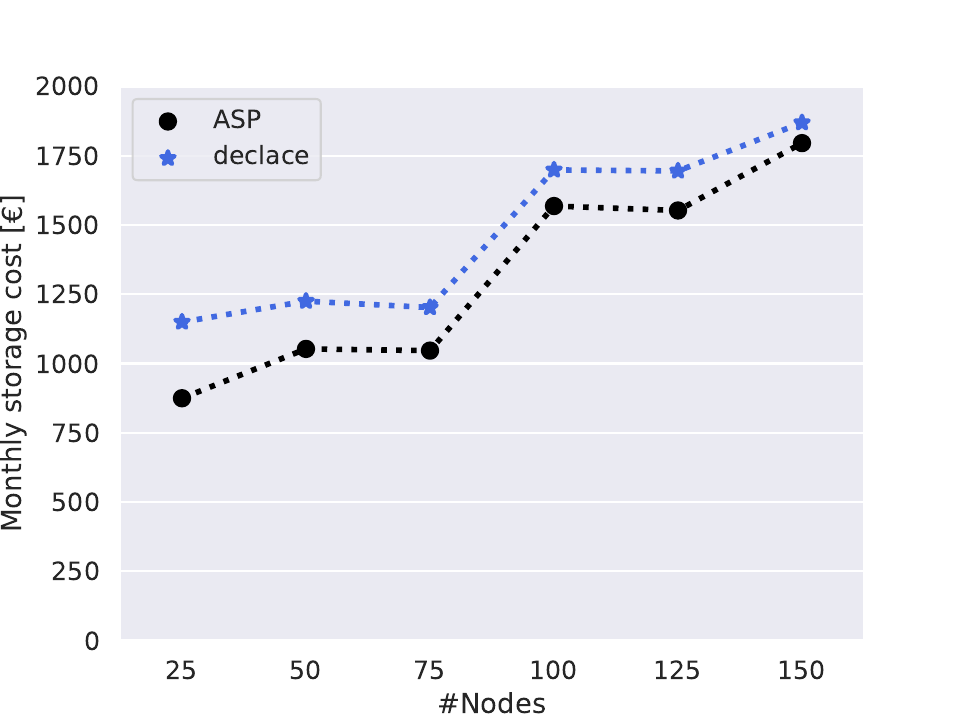}
  \caption{costs}
  \label{fig:d12images_cost}
\end{subfigure}
\caption{\proimg versus ASP at varying number of nodes (12 images).}
\label{fig:d12images}
\end{figure*}

%% file: src/relatedwork.tex
\vspace{-0.5cm}
\section{Related Work}
\label{sec:related}

% Cloud-edge applications are oftentimes latency-sensitive and require to run onto IoT-enabled edge nodes, which in turn could experiment unstable connectivity to centralised container \rev{registries}~\cite{DBLP:journals/access/NahaGGJGXR18}. These settings call for methodologies to plan (and re-plan) container image placement onto distributed cloud-edge repositories~\cite{DBLP:conf/icccn/DarrousLI19} to reduce application start-up times and enable rapid application deployment even in case of limited edge connectivity or unpredictable network dynamics. 

Despite being interesting and challenging, the problem of distributing container images in Cloud-Edge landscapes has only been investigated in a few recent works, viz.,~\cite{DBLP:conf/ipccc/Becker0K21, DBLP:conf/icccn/DarrousLI19,DBLP:conf/europar/GazzettiRKC17,CLOUD2023,zhang2021container}, none of which considers runtime image placement adaptation.
Indeed, most efforts have been devoted to the (online or offline) problem of placing multiservice applications onto specific target nodes in a context- and QoS-aware manner as surveyed, for instance, in~\cite{DBLP:journals/csur/MahmudRB20}. Among these, Forti \textit{et al.}~\cite{DBLP:journals/fgcs/FortiGB21} and Herrera \textit{et al.}~\cite{DBLP:journals/computing/HerreraBFBM23}   proposed to exploit  continuous reasoning to speed up decision-making at runtime, when selecting suitable placements for multiservice applications. %Even the recent survey by Kaur \textit{et al.}~\cite{DBLP:journals/ijnm/KaurGS22} on container placement and migration strategies for cloud-edge settings does not mention existing solutions for container image placement. 

Focussing on Docker image distribution in cloud datacentres, Kangjin \textit{et al.}~\cite{DBLP:conf/saso/KangjinYYHL17} proposed a peer-to-peer extension of the Docker Image Registry system. 
Similarly, yet targeting cloud-edge computing settings, the authors of~\cite{DBLP:conf/ipccc/Becker0K21} and \cite{DBLP:conf/europar/GazzettiRKC17} devise fully decentralised container registries based on local agents compliant with the Docker registry API. Also Zhang \textit{et al.}~\cite{zhang2021container} propose a container image distribution architecture based on peer-to-peer pre-caching of images onto alive edge nodes. Although missing QoS-, cost- and context-aware strategies to decide on container image placement, the above works show that distributing images around the network can lead to substantially lower container pull times, improved image availability, and reduced network traffic. \rev{They might constitute the enabling technology for \proimg{} to distribute images.}

Darrous \textit{et al.}~\cite{DBLP:conf/icccn/DarrousLI19} first investigated the problem of QoS-aware image placement in cloud-edge settings by proposing two placement strategies based on k-center optimisation. While their problem model is similar to ours, they do not consider storage costs and only focus on reducing container retrieval times. They employ a Python-based solver for k-center problems and propose sorting images to be placed by decreasing size as we do. They assess their proposal against a best-fit and a random strategy over networks up to 50 nodes.
Finally, Theodoropoulos \textit{et al.}~\cite{CLOUD2023} formulate the image placement problem as a minimum vertex cover problem and solve it through a reinforcement learning approach. They assess their proposal against a greedy strategy. % placing images onto infrastructures made of at most 512 nodes. 
Execution times of the proposed reinforcement strategies are two orders of magnitude slower than those of the greedy strategy but obtain on average a better vertex cover. 

Differently from \proimg{}, both~\cite{DBLP:conf/icccn/DarrousLI19} and \cite{CLOUD2023} focus on initial image placement and do not consider the possibility of adapting such placement in response to runtime infrastructural or image repository changes, which might turn useful in ever-changing cloud-edge landscapes, working in continuity with CI/CD release pipelines. 
To the best of our knowledge, our work is the first employing continuous reasoning for adaptively solving the cloud-edge image placement problem at runtime while reducing execution times, by focussing only on those images that need to be re-allocated. Thanks to its declarative nature, our approach features better readability, maintainability and extensibility than procedural solutions like \cite{DBLP:conf/icccn/DarrousLI19}, and better 
%explainability, 
interpretability and incrementality of the obtained solutions, which \rev{can be} difficult to obtain when relying on learning approaches, like \cite{CLOUD2023}.

%% file: src/conclusions.tex
\section{Concluding Remarks}
\label{sec:conclusions}

%In this article, we have studied the adaptive placement of container images in cloud-edge settings to reduce image provisioning times and, consequently, application start-up times, also in presence of infrastructure and image requirements' changes. 

\rev{In this article, we have formally modelled the (optimal and continuous versions of the) problem of adaptive placement of container images in cloud-edge setting and proved its decisional version is NP-hard.} Then, we have proposed a declarative pipeline to solve such a problem in a QoS-~, context- and cost-aware manner. We have implemented our methodology through logic programming in the open-source prototype \proimg{}, which combines ASP to determine cost-optimal initial placements and Prolog for adapting image placements at runtime via continuous reasoning. We have thoroughly assessed our prototype by simulating its functioning over lifelike data.
The incremental approach of continuous reasoning reduces the size of the problem instances to be solved and the number of image migrations. This makes runtime decision-making faster, also avoiding unneeded data transfers. The synergistic use of Prolog and ASP balances initial cost optimisation with prompt runtime adaptation. The usage of logic programming to define our methodology makes it more concise ($\simeq$ 100 lines of code), readable and extensible than imperative solutions. 
Future work on this line may include:

\begin{itemize}
    \item the \textit{extension} of \proimg{} \rev{to account for image layers, dependencies between non-base images, and associated transfer delays,} and with with probabilistic reasoning~\cite{AzzRigLam21-AIJ-IJ} to assess candidate placements against the uncertainty of infrastructure variations, so to improve the resilience of found placements against such variations (as in \cite{DBLP:journals/tplp/FortiPB22}),
    \item the \textit{combination} of our methodology \textit{with deep learning} solutions for container image placement (like~\cite{CLOUD2023}) to extend them with interpretability, explainability and incrementality features and to equip \proimg{} with prediction features, eventually improving decisions in the spirit on neuro-symbolic artificial intelligence~\cite{neuroceocosymbAI}, and
    \item further experimental assessment over real data about the dynamics of cloud-edge networks (currently unavailable) and the integration of \proimg with a \textit{distributed container repository} to be assessed in \textit{testbed} settings with real images to manage.
\end{itemize}

%% file: main.bbl

\begin{thebibliography}{30}
\ifx \bisbn   \undefined \def \bisbn  #1{ISBN #1}\fi
\ifx \binits  \undefined \def \binits#1{#1}\fi
\ifx \bauthor  \undefined \def \bauthor#1{#1}\fi
\ifx \batitle  \undefined \def \batitle#1{#1}\fi
\ifx \bjtitle  \undefined \def \bjtitle#1{#1}\fi
\ifx \bvolume  \undefined \def \bvolume#1{\textbf{#1}}\fi
\ifx \byear  \undefined \def \byear#1{#1}\fi
\ifx \bissue  \undefined \def \bissue#1{#1}\fi
\ifx \bfpage  \undefined \def \bfpage#1{#1}\fi
\ifx \blpage  \undefined \def \blpage #1{#1}\fi
\ifx \burl  \undefined \def \burl#1{\textsf{#1}}\fi
\ifx \doiurl  \undefined \def \doiurl#1{\url{https://doi.org/#1}}\fi
\ifx \betal  \undefined \def \betal{\textit{et al.}}\fi
\ifx \binstitute  \undefined \def \binstitute#1{#1}\fi
\ifx \binstitutionaled  \undefined \def \binstitutionaled#1{#1}\fi
\ifx \bctitle  \undefined \def \bctitle#1{#1}\fi
\ifx \beditor  \undefined \def \beditor#1{#1}\fi
\ifx \bpublisher  \undefined \def \bpublisher#1{#1}\fi
\ifx \bbtitle  \undefined \def \bbtitle#1{#1}\fi
\ifx \bedition  \undefined \def \bedition#1{#1}\fi
\ifx \bseriesno  \undefined \def \bseriesno#1{#1}\fi
\ifx \blocation  \undefined \def \blocation#1{#1}\fi
\ifx \bsertitle  \undefined \def \bsertitle#1{#1}\fi
\ifx \bsnm \undefined \def \bsnm#1{#1}\fi
\ifx \bsuffix \undefined \def \bsuffix#1{#1}\fi
\ifx \bparticle \undefined \def \bparticle#1{#1}\fi
\ifx \barticle \undefined \def \barticle#1{#1}\fi
\bibcommenthead
\ifx \bconfdate \undefined \def \bconfdate #1{#1}\fi
\ifx \botherref \undefined \def \botherref #1{#1}\fi
\ifx \url \undefined \def \url#1{\textsf{#1}}\fi
\ifx \bchapter \undefined \def \bchapter#1{#1}\fi
\ifx \bbook \undefined \def \bbook#1{#1}\fi
\ifx \bcomment \undefined \def \bcomment#1{#1}\fi
\ifx \oauthor \undefined \def \oauthor#1{#1}\fi
\ifx \citeauthoryear \undefined \def \citeauthoryear#1{#1}\fi
\ifx \endbibitem  \undefined \def \endbibitem {}\fi
\ifx \bconflocation  \undefined \def \bconflocation#1{#1}\fi
\ifx \arxivurl  \undefined \def \arxivurl#1{\textsf{#1}}\fi
\csname PreBibitemsHook\endcsname

%%% 1
\bibitem[\protect\citeauthoryear{Naha et~al.}{2018}]{DBLP:journals/access/NahaGGJGXR18}
\begin{barticle}
\bauthor{\bsnm{Naha}, \binits{R.K.}}, \betal:
\batitle{Fog computing: Survey of trends, architectures, requirements, and research directions}.
\bjtitle{{IEEE} Access}
\bvolume{6},
\bfpage{47980}--\blpage{48009}
(\byear{2018})
\end{barticle}
\endbibitem

%%% 2
\bibitem[\protect\citeauthoryear{Khan et~al.}{2019}]{DBLP:journals/fgcs/KhanAHYA19}
\begin{barticle}
\bauthor{\bsnm{Khan}, \binits{W.Z.}}, \betal:
\batitle{Edge computing: {A} survey}.
\bjtitle{Future Gener. Comput. Syst.}
\bvolume{97},
\bfpage{219}--\blpage{235}
(\byear{2019})
\end{barticle}
\endbibitem

%%% 3
\bibitem[\protect\citeauthoryear{Tortonesi et~al.}{2019}]{DBLP:journals/fgcs/TortonesiGMRSS19}
\begin{barticle}
\bauthor{\bsnm{Tortonesi}, \binits{M.}}, \betal:
\batitle{Taming the iot data deluge: An innovative information-centric service model for fog computing applications}.
\bjtitle{Future Gener. Comput. Syst.}
\bvolume{93},
\bfpage{888}--\blpage{902}
(\byear{2019})
\end{barticle}
\endbibitem

%%% 4
\bibitem[\protect\citeauthoryear{A{\"{\i}}t{-}Salaht et~al.}{2021}]{DBLP:journals/csur/Ait-SalahtDL20}
\begin{barticle}
\bauthor{\bsnm{A{\"{\i}}t{-}Salaht}, \binits{F.}},
\bauthor{\bsnm{Desprez}, \binits{F.}},
\bauthor{\bsnm{Lebre}, \binits{A.}}:
\batitle{An overview of service placement problem in fog and edge computing}.
\bjtitle{{ACM} Comput. Surv.}
\bvolume{53}(\bissue{3}),
\bfpage{65}--\blpage{16535}
(\byear{2021})
\end{barticle}
\endbibitem

%%% 5
\bibitem[\protect\citeauthoryear{Forti et~al.}{2022}]{DBLP:journals/logcom/FortiBB22}
\begin{barticle}
\bauthor{\bsnm{Forti}, \binits{S.}}, \betal:
\batitle{Declarative continuous reasoning in the cloud-iot continuum}.
\bjtitle{J. Log. Comput.}
\bvolume{32}(\bissue{2}),
\bfpage{206}--\blpage{232}
(\byear{2022})
\end{barticle}
\endbibitem

%%% 6
\bibitem[\protect\citeauthoryear{Mahmud et~al.}{2021}]{DBLP:journals/csur/MahmudRB20}
\begin{barticle}
\bauthor{\bsnm{Mahmud}, \binits{M.R.}}, \betal:
\batitle{Application management in fog computing environments: {A} taxonomy, review and future directions}.
\bjtitle{{ACM} Comput. Surv.}
\bvolume{53}(\bissue{4}),
\bfpage{88}--\blpage{18843}
(\byear{2021})
\end{barticle}
\endbibitem

%%% 7
\bibitem[\protect\citeauthoryear{Shakarami et~al.}{2022}]{DBLP:journals/jsa/ShakaramiSGNM22}
\begin{barticle}
\bauthor{\bsnm{Shakarami}, \binits{A.}}, \betal:
\batitle{Resource provisioning in edge/fog computing: {A} comprehensive and systematic review}.
\bjtitle{J. Syst. Archit.}
\bvolume{122},
\bfpage{102362}
(\byear{2022})
\end{barticle}
\endbibitem

%%% 8
\bibitem[\protect\citeauthoryear{Soldani et~al.}{2018}]{DBLP:journals/jss/SoldaniTH18}
\begin{barticle}
\bauthor{\bsnm{Soldani}, \binits{J.}}, \betal:
\batitle{The pains and gains of microservices: {A} systematic grey literature review}.
\bjtitle{J. Syst. Softw.}
\bvolume{146},
\bfpage{215}--\blpage{232}
(\byear{2018})
\end{barticle}
\endbibitem

%%% 9
\bibitem[\protect\citeauthoryear{Darrous et~al.}{2019}]{DBLP:conf/icccn/DarrousLI19}
\begin{bchapter}
\bauthor{\bsnm{Darrous}, \binits{J.}}, \betal:
\bctitle{On the importance of container image placement for service provisioning in the edge}.
In: \bbtitle{{IEEE} {ICCCN}},
pp. \bfpage{1}--\blpage{9}
(\byear{2019})
\end{bchapter}
\endbibitem

%%% 10
\bibitem[\protect\citeauthoryear{Theodoropoulos et~al.}{2023}]{CLOUD2023}
\begin{bchapter}
\bauthor{\bsnm{Theodoropoulos}, \binits{T.}}, \betal:
\bctitle{{GNOSIS:} proactive image placement using graph neural networks {\&} deep reinforcement learning}.
In: \bbtitle{{IEEE} {CLOUD}},
pp. \bfpage{120}--\blpage{128}
(\byear{2023})
\end{bchapter}
\endbibitem

%%% 11
\bibitem[\protect\citeauthoryear{Aslanpour et~al.}{2021}]{DBLP:conf/acsw/AslanpourTCJSTA21}
\begin{bchapter}
\bauthor{\bsnm{Aslanpour}, \binits{M.S.}}, \betal:
\bctitle{Serverless edge computing: Vision and challenges}.
In: \bbtitle{{ACM} {ACSW}},
pp. \bfpage{10}--\blpage{11010}
(\byear{2021})
\end{bchapter}
\endbibitem

%%% 12
\bibitem[\protect\citeauthoryear{Forti et~al.}{2021}]{DBLP:journals/fgcs/FortiGB21}
\begin{barticle}
\bauthor{\bsnm{Forti}, \binits{S.}}, \betal:
\batitle{Lightweight self-organising distributed monitoring of fog infrastructures}.
\bjtitle{Future Gener. Comput. Syst.}
\bvolume{114},
\bfpage{605}--\blpage{618}
(\byear{2021})
\end{barticle}
\endbibitem

%%% 13
\bibitem[\protect\citeauthoryear{Springer}{}]{dockerupdates}
\begin{botherref}
\oauthor{\bsnm{Springer}, \binits{A.}}:
{How often are Docker images updated - Revisited}.
\url{https://anchore.com/blog/how-often-are-docker-images-updated-revisited/}.
Last accessed: \today
\end{botherref}
\endbibitem

%%% 14
\bibitem[\protect\citeauthoryear{Distefano et~al.}{2019}]{DBLP:journals/cacm/DistefanoFLO19}
\begin{barticle}
\bauthor{\bsnm{Distefano}, \binits{D.}}, \betal:
\batitle{Scaling static analyses at facebook}.
\bjtitle{Commun. {ACM}}
\bvolume{62}(\bissue{8}),
\bfpage{62}--\blpage{70}
(\byear{2019})
\end{barticle}
\endbibitem

%%% 15
\bibitem[\protect\citeauthoryear{Herrera et~al.}{2023}]{DBLP:journals/computing/HerreraBFBM23}
\begin{barticle}
\bauthor{\bsnm{Herrera}, \binits{J.L.}}, \betal:
\batitle{Continuous qos-aware adaptation of cloud-iot application placements}.
\bjtitle{Computing}
\bvolume{105}(\bissue{9}),
\bfpage{2037}--\blpage{2059}
(\byear{2023})
\end{barticle}
\endbibitem

%%% 16
\bibitem[\protect\citeauthoryear{Garey and Johnson}{1979}]{DBLP:books/fm/GareyJ79}
\begin{botherref}
\oauthor{\bsnm{Garey}, \binits{M.R.}},
\oauthor{\bsnm{Johnson}, \binits{D.S.}}:
{Computers and Intractability: {A} Guide to the Theory of NP-Completeness}
(1979)
\end{botherref}
\endbibitem

%%% 17
\bibitem[\protect\citeauthoryear{Sterling and Shapiro}{1994}]{sterling1994art}
\begin{bbook}
\bauthor{\bsnm{Sterling}, \binits{L.}},
\bauthor{\bsnm{Shapiro}, \binits{E.}}:
\bbtitle{The Art of Prolog: Advanced Programming Techniques}.
\bsertitle{Logic programming}.
\bpublisher{{MIT} Press},
\blocation{USA}
(\byear{1994})
\end{bbook}
\endbibitem

%%% 18
\bibitem[\protect\citeauthoryear{Lifschitz}{2019}]{DBLP:books/sp/Lifschitz19}
\begin{bbook}
\bauthor{\bsnm{Lifschitz}, \binits{V.}}:
\bbtitle{Answer Set Programming}.
\bpublisher{Springer},
\blocation{Cham}
(\byear{2019})
\end{bbook}
\endbibitem

%%% 19
\bibitem[\protect\citeauthoryear{Russell and Norvig}{2020}]{AIMA}
\begin{bbook}
\bauthor{\bsnm{Russell}, \binits{S.}},
\bauthor{\bsnm{Norvig}, \binits{P.}}:
\bbtitle{Artificial Intelligence: {A} Modern Approach (4th Edition)}.
\bpublisher{Pearson},
\blocation{USA}
(\byear{2020})
\end{bbook}
\endbibitem

%%% 20
\bibitem[\protect\citeauthoryear{Potassco}{2023}]{potasscoGuide}
\begin{botherref}
\oauthor{\bsnm{Potassco}}:
Potassco User Guide
(2023)
\end{botherref}
\endbibitem

%%% 21
\bibitem[\protect\citeauthoryear{Barab{\'a}si and Albert}{1999}]{barabasi1999emergence}
\begin{barticle}
\bauthor{\bsnm{Barab{\'a}si}, \binits{A.-L.}},
\bauthor{\bsnm{Albert}, \binits{R.}}:
\batitle{Emergence of scaling in random networks}.
\bjtitle{Science}
\bvolume{286}(\bissue{5439}),
\bfpage{509}--\blpage{512}
(\byear{1999})
\end{barticle}
\endbibitem

%%% 22
\bibitem[\protect\citeauthoryear{Hagberg et~al.}{2008}]{hagberg2008networkx}
\begin{bchapter}
\bauthor{\bsnm{Hagberg}, \binits{A.A.}}, \betal:
\bctitle{Exploring network structure, dynamics, and function using {NetworkX}}.
In: \beditor{\bsnm{Varoquaux}, \binits{G.}},
\beditor{\bsnm{Vaught}, \binits{T.}},
\beditor{\bsnm{Millman}, \binits{J.}} (eds.)
\bbtitle{Proceedings of the 7th Python in Science Conference},
pp. \bfpage{11}--\blpage{15}
(\byear{2008})
\end{bchapter}
\endbibitem

%%% 23
\bibitem[\protect\citeauthoryear{Barabási et~al.}{2000}]{BARABASIVABENEPERLERETI}
\begin{barticle}
\bauthor{\bsnm{Barabási}, \binits{A.-L.}},
\bauthor{\bsnm{Albert}, \binits{R.}},
\bauthor{\bsnm{Jeong}, \binits{H.}}:
\batitle{Scale-free characteristics of random networks: the topology of the world-wide web}.
\bjtitle{Physica A: Statistical Mechanics and its Applications}
\bvolume{281}(\bissue{1}),
\bfpage{69}--\blpage{77}
(\byear{2000})
\end{barticle}
\endbibitem

%%% 24
\bibitem[\protect\citeauthoryear{Becker et~al.}{2021}]{DBLP:conf/ipccc/Becker0K21}
\begin{bchapter}
\bauthor{\bsnm{Becker}, \binits{S.}}, \betal:
\bctitle{{EdgePier}: {P2P}-based container image distribution in edge computing environments}.
In: \bbtitle{{IEEE} {IPCCC}},
pp. \bfpage{1}--\blpage{8}
(\byear{2021})
\end{bchapter}
\endbibitem

%%% 25
\bibitem[\protect\citeauthoryear{Gazzetti et~al.}{2017}]{DBLP:conf/europar/GazzettiRKC17}
\begin{botherref}
\oauthor{\bsnm{Gazzetti}, \binits{M.}}, et al.:
Scalable linux container provisioning in fog and edge computing platforms
\textbf{10659},
304--315
(2017)
\end{botherref}
\endbibitem

%%% 26
\bibitem[\protect\citeauthoryear{Zhang et~al.}{2021}]{zhang2021container}
\begin{bchapter}
\bauthor{\bsnm{Zhang}, \binits{J.}}, \betal:
\bctitle{A container deployment optimization method for edge computing}.
In: \bbtitle{Journal of Physics: Conference Series},
vol. \bseriesno{1927},
p. \bfpage{012017}
(\byear{2021}).
\bcomment{IOP Publishing}
\end{bchapter}
\endbibitem

%%% 27
\bibitem[\protect\citeauthoryear{Wang et~al.}{2017}]{DBLP:conf/saso/KangjinYYHL17}
\begin{botherref}
\oauthor{\bsnm{Wang}, \binits{K.}}, et al.:
{FID:} {A} faster image distribution system for docker platform,
191--198
(2017)
\end{botherref}
\endbibitem

%%% 28
\bibitem[\protect\citeauthoryear{Azzolini et~al.}{2021}]{AzzRigLam21-AIJ-IJ}
\begin{barticle}
\bauthor{\bsnm{Azzolini}, \binits{D.}},
\bauthor{\bsnm{Riguzzi}, \binits{F.}},
\bauthor{\bsnm{Lamma}, \binits{E.}}:
\batitle{A semantics for hybrid probabilistic logic programs with function symbols}.
\bjtitle{Artificial Intelligence}
\bvolume{294},
\bfpage{103452}
(\byear{2021})
\end{barticle}
\endbibitem

%%% 29
\bibitem[\protect\citeauthoryear{Forti et~al.}{2022}]{DBLP:journals/tplp/FortiPB22}
\begin{barticle}
\bauthor{\bsnm{Forti}, \binits{S.}}, \betal:
\batitle{{Probabilistic QoS-aware Placement of {VNF} Chains at the Edge}}.
\bjtitle{Theory Pract. Log. Program.}
\bvolume{22}(\bissue{1}),
\bfpage{1}--\blpage{36}
(\byear{2022})
\end{barticle}
\endbibitem

%%% 30
\bibitem[\protect\citeauthoryear{d’Avila Garcez and Lamb}{2023}]{neuroceocosymbAI}
\begin{barticle}
\bauthor{\bsnm{Garcez}, \binits{A.}},
\bauthor{\bsnm{Lamb}, \binits{L.C.}}:
\batitle{Neurosymbolic ai: the 3rd wave}.
\bjtitle{Artif. Intell. Rev.}
\bvolume{56},
\bfpage{12387}--\blpage{12406}
(\byear{2023})
\end{barticle}
\endbibitem

\end{thebibliography}
